\documentclass[prx,twocolumn,aps,superscriptaddress,nofootinbib,10pt,longbibliography]{revtex4-2}

\usepackage[utf8]{inputenc}
\usepackage{needspace}
\usepackage[english]{babel}
\usepackage{graphicx}
\usepackage{dsfont}
\usepackage{amsmath, amssymb}
\usepackage{amsthm}
\usepackage{amsfonts}
\usepackage{physics}
\usepackage{bm}
\usepackage[dvipsnames]{xcolor}
\usepackage{subcaption}
\usepackage{tikz}
\usetikzlibrary{matrix, arrows.meta,positioning, decorations.pathreplacing}

\usepackage[colorlinks=true, linkcolor=Purple,citecolor=Violet,urlcolor=Violet]{hyperref}

\usepackage{tabularx}
\usepackage[noend,linesnumbered, ruled, vlined]{algorithm2e}
\let\savednl\nl
\newcommand{\nonl}{\renewcommand{\nl}{\let\nl\savednl}}
\usepackage{tcolorbox}
\tcbuselibrary{breakable}

\newtheorem{theorem}{Theorem}
\newtheorem{lemma}{Lemma}
\newtheorem{proposition}{Proposition}
\newtheorem*{proposition*}{Proposition}
\newtheorem{definition}{Definition}
\newtheorem{corollary}{Corollary}
\theoremstyle{remark}
\newtheorem{remark}{Remark}

\newcommand{\bZ}{\mathbb{Z}}
\newcommand{\bR}{\mathbb{R}}
\newcommand{\one}{\mathds{1}}

\newcommand{\Sym}{\mathrm{Sym}}
\newcommand{\Imp}{\mathrm{Imp}}
\newcommand{\DecX}{\mathrm{Decode}_X}
\newcommand{\DecZ}{\mathrm{Decode}_Z}
\newcommand{\CheckCondC}{\mathrm{CheckCondC}}
\newcommand{\GenTwisted}{\mathrm{GenTwisted}}
\newcommand{\GenLinking}{\mathrm{GenLinking}}
\DeclareMathOperator{\im}{im}

\DeclareMathOperator{\Prob}{Prob}

\newcommand{\myparagraph}[1]{\needspace{12pt}
\par
\vspace{0.5\baselineskip}\noindent\textbf{#1}\;--\;}

\begin{document}

\title{Efficiently estimating failure rates of fault-tolerant logical non-Clifford blocks}
\author{Julio C. Magdalena de la Fuente}
\email{jm@juliomagdalena.de}
\affiliation{\footnotesize Dahlem Center for Complex Quantum Systems, Freie Universit\"at Berlin, 14195 Berlin, Germany}
\affiliation{\footnotesize School of Physics, The University of Sydney, NSW 2006, Australia}
\author{Thomas R. Scruby}
\affiliation{\footnotesize Okinawa Institute of Science and Technology, 1919-1 Tancha, Onna, Okinawa, Japan}
\affiliation{\footnotesize Iceberg Quantum}

\begin{abstract}
    Useful quantum computation requires the fault-tolerant implementations of a universal gate set which are generically not efficiently simulable.
    We devise an algebraic framework that allows for efficient sampling from the measurement distribution of fault-tolerant circuits that implement diagonal logic gates in the third level of the Clifford hierarchy.
    Upon successful sampling we also provide sufficient conditions for decoding success.
    The resulting estimate for the logical failure rate is an overestimate, which becomes more accurate for blocks that are fault tolerant against arbitrary local errors.
    The non-Clifford simulation overhead is independent of the number of logical qubits, making the method particularly attractive for large-scale simulations.
    The framework is based on viewing the non-Clifford gates in the circuit as a cohomology invariant of an underlying spacetime fault complex.
    The method can also be interfaced with Clifford simulators to simulate larger fault-tolerant circuits and algorithmic subroutines.
\end{abstract}

\maketitle

\section{Introduction}
Large-scale fault-tolerant quantum computation offers speed-ups for important problems in quantum simulation and cryptography~\cite{Shor1994Algorithms, Daley2022}.
In recent years, significant progress has been made towards building a useful quantum computer and small-scale experiments demonstrated several building blocks of fault-tolerant quantum algorithms~\cite{Postler2022, Bluvstein2023, Google2024, Gupta2024MagicState, Zhou2025, Rosenfeld2025Cultivation}.
Given the rapid progress in experimental implementation, reliable methods to  efficiently estimate the performance of logical gadgets and important algorithmic subroutines are needed.

Many parts of fault-tolerant algorithms use Clifford circuits, which can be simulated efficiently using the \emph{stabilizer formalism}~\cite{gottesman1997phd, Aaronson2004circuits}.
Universal quantum computation also requires \emph{non-Clifford operations}, which generally cannot be simulated efficiently.
However, fault-tolerant \emph{non-Clifford blocks}, components that reliably implement non-Clifford logic in the presence of noise, often have useful structure.
In particular, several approaches rely on diagonal physical non-Clifford gates~\cite{Cui2017diagonal}.
These properties may enable efficient, scalable tools for estimating their performance.

In this work, we introduce a new method to estimate logical failure rates of fault-tolerant non-Clifford blocks composed of CSS operations and diagonal gates in the third level of the Clifford hierarchy~\cite{Gottesman1999}.
These classes of circuits include scalable magic-state preparation schemes that rely on code-switching~\cite{Paetznick2013, Anderson2014, bombin2015gaugecolor, bombin2016jump} or gauging diagonal Clifford symmetries of a CSS code~\cite{tiedinknots, bauer2025planar, christos2026gauging, zhu2026nonabelianqldpc, williamson2026fastmagicstatepreparation} as well as unitary logical non-Clifford gates implemented by fault-tolerant unitary gates on a CSS code~\cite{EastinKnill, BravyiKoenig}, or through code-deformation operations that involve non-Abelian stabilizer codes~\cite{tiedinknots, huang2026hybridlattice, manjunath2026groupsurface, warman2026constantdepth}.

Our method relies on the observation that the measurement distribution is highly constrained if the circuit is fault-tolerant against arbitrary local errors.
By comparing to a \emph{clean} circuit, without any errors and non-trivial measurement outcomes, we can characterize the syndrome distribution via modifications to the clean distribution.
This allows for an efficient simulation of the decoding problem.
Moreover, we identify a sufficient condition for decoding success without tracking the logical state explicitly.
This makes the method particularly useful for non-Clifford blocks that act on many logical qubits in parallel~\cite{zhu2026nonabelianqldpc}.

\myparagraph{Relation to previous works}
The mathematical formalism developed in this work is the closest to Ref.~\cite{bombin2018propagation}, where the correlated structure on $Z$ errors due to residual $X$ errors was highlighted for the 3D color code.
This was further formalised in Ref.~\cite{bauer2025planar}.
Ref.~\cite{highthresholdJIT} applied the same approach to circuits that measure the stabilizers of a 2D non-Abelian code and sketched the steps in generalizing the approach to general spacetime LDPC protocols.
Here, we make this generalization explicit and extend it to make it usable in practice with Pauli-frame based Clifford simulators, which we describe further in the companion paper~\cite{companion}.

Other works have studied non-Pauli errors in error-correcting codes that can emerge from propagating Pauli errors through non-Clifford circuits~\cite{Scruby2022nonPauli} and implemented their simulation~\cite{Scruby2022JIT}.
Recently, Ref.~\cite{Surti2026simulation} uses error propagation to simulate circuits that prepare encoded magic states.
It requires the propagated errors to be Clifford but cannot simulate deep circuits with repeated rounds of syndrome extraction efficiently as they require the errors to be propagated to the end.\\

\emph{Note added:} Towards the end of finishing the work for this paper we became aware of parallel related work by Takada, Bartlett and Williamson~\cite{yugopaper} that introduces a Clifford-stabilizer tableau to efficiently simulate large classes of non-Clifford circuits and magic-state preparation protocols under Pauli noise.
The main difference to our work is that their scheme is based on an exact representation of the logical state and our method focusses on simulating the decoding problem only.
The scheme by Takada et al. requires a more elaborate tableau simulation but also enables the simulation of protocols with a low fault-distance that rely on post-selection techniques~\cite{psiCultivation, gidney2024magicstatecultivation}.
In practice, both approaches could be complementary, as our approach is suited to large-scale decoding simulations, and is likely more scalable for structured circuits, see for example the simulations in Ref.~\cite{highthresholdJIT}.
It would be interesting to explore the details in practical performance of both methods more, after optimizing important subroutines to make a more reliable comparison.

\section{Problem sketch and overview of the results}
A simulation to estimate failure rates of an error-correcting protocol under stochastic noise needs to go through the following steps:
\begin{enumerate}
    \item \textbf{Sample an error} configuration: The stochastic noise model is defined by a probability distribution of errors.
    \item \textbf{Sample a measurement} outcome: Quantum-mechanical measurements are probabilistic and in each run the circuit projects onto a fixed trajectory, governed by a probability distribution that needs to be sampled from.
    \item \textbf{Decode} the observed outcome: Feed the sampled outcome to the decoder which determines a correction, which is applied into the circuit, or tracked in software.
    \item \textbf{Decide} on decoding success/failure: Given a valid correction it must be decided if the errors have been corrected correctly or not.
\end{enumerate}
For general circuits steps 2 and 4 can lead to inefficiencies in the simulation compared to Clifford circuits as they might require a full state-vector simulation.
An additional complication arises from the fact that the logical error channel induced by stochastic errors in a general circuit is not necessarily stochastic on the logical level, due to unpredictable effects of projective measurements.
In this work, we focus on a special class of circuits for which all steps in the list above can be done efficiently, without needing to keep track of any intermediate representation of the actual wave function.

\begin{tcolorbox}[
  colback=white!95!black,
  colframe=white!20!black, 
  boxrule=0.5pt,
  arc=2mm,
  boxsep=0pt
]
We consider \emph{third-order circuits}. These are circuits composed of
\begin{itemize}
    \item CSS operations: controlled-$X$ gates, Pauli-$X$ and Pauli-$Z$ measurements,
    \item diagonal gates in the third level of the Clifford hierarchy: $T$, controlled-$S$ gates and controlled-controlled-$Z$ gates,
    \item Pauli errors and corrections.
\end{itemize}
The third-order circuits that we focus on define \emph{fault-tolerant logical blocks}.
These are circuits that implement a logical third-order quantum instrument from an input CSS code into an output CSS code and are designed to be fault tolerant against arbitrary local errors.
\end{tcolorbox}

Since all $Z$ operators commute with the non-Clifford gates, the distribution of $Z$ measurements is easy to sample from.
The difficulty for both step 2 and 4 above hence lies in sampling from and decoding the $X$ measurement outcomes.
\begin{tcolorbox}[
  colback=white!95!black,   
  colframe=white!20!black, 
  boxrule=0.5pt,
  arc=2mm,
  boxsep=0pt,
  halign=center
]
    We can exactly sample from the measurement distribution of a third-order circuit when the $X$-like errors in the circuit together with the corrections fulfil a \emph{contractibility condition} (condition C).\\[6pt]
    In this case we can also give a precise condition on when the decoding of $X$-measurement outcomes succeeds or fails.
\end{tcolorbox}
In a circuit that is fault tolerant against all local errors, condition C requires the regions of the combined $X$-like errors in the circuit to be sufficiently well-isolated from each other.
In Sec.~\ref{sec:FT-thirdorder} we give a precise technical definition.
The origin of this condition is that our simulation paradigm requires that the measurement distribution is independent of the logical information which cannot be guaranteed if condition C is not fulfilled.

Under stochastic Pauli noise not every sampled instance will fulfil condition C.
In these cases we declare decoding failure.
In general, it is also exponentially hard in the volume of the circuit to check condition C exactly. Thus for an efficient simulator we require an approximate method for checking condition C or sufficient structure to simplify the problem. 
Topological protocols are paradigmatic examples of instances with more structure and Ref.~\cite{highthresholdJIT} showcased how our formalism can be applied for a precise estimate.
For the general case where we check the condition approximately, we will overestimate the failure rate.
In practice, we expect the instances in which condition C breaks to be rare in the fault tolerant setting and the cases in which it would be misidentified by an approximate condition to correlate with low confidence in the decoding output.

\myparagraph{Mathematical structure of the formalism}
Using a homological description of third-order circuits we derive a \emph{non-Clifford detector-error model}.
As in a CSS circuit, the $X$ errors $b_e$ and $Z$ measurement outcomes $b_m$ are subject to linear constraints called \emph{$Z$-detectors}, $D_Z$, that can be captured by a single linear map $d:T_X\to D_Z$, where $T_X$ is the space of $X$-like error locations.
The $X$ correction $b_c$ is chosen such that
\begin{align}
    d (b_m+b_e)= db_c,
\end{align}
with the goal that the combined error $b=b_m+b_c+b_e$ has a low overlap region with the non-Clifford gates in the circuit.

Given that the combined error $b$ fulfils condition C, we derive linear constraints called \emph{twisted $X$-detectors}, $D_X^b$, on $Z$ errors $c_e$ and $X$ measurement outcomes $c_m$.
These are captured by a single linear map $\partial^b\colon T_Z\to D_X^b$, where $T_Z$ is the space of $Z$-like error locations, and an affine shift $\kappa_b\in D_X^b$. Then
\begin{align}\label{eq:intro-twisted-constraints}
    \partial^b c_e = \partial^b c_m + \kappa_b
\end{align}
defines the exact decoding problem for $Z$ errors.
The key difference to the usual detector error model in Clifford circuits is that the linear constraints explicitly depend on $b$.
\begin{tcolorbox}[
  colback=white!95!black,   
  colframe=white!20!black, 
  boxrule=0.5pt,
  arc=2mm,
  boxsep=0pt,
  halign=center
]
We show that every measurement outcome in the subspace defined by Eq.\eqref{eq:intro-twisted-constraints} appears with the same probability.
\end{tcolorbox}

We construct $\partial^b$ from knowledge about the constraints for $b=0$.
Fault tolerance for the $b=0$ instances requires the non-Clifford gates to obey a very strong symmetry with respect to the detecting regions in the underlying CSS circuit.
\begin{tcolorbox}[
  colback=white!95!black,   
  colframe=white!20!black, 
  boxrule=0.5pt,
  arc=2mm,
  boxsep=0pt,
  halign=center
]
    The non-Clifford gates in the circuit must define a cohomology invariant of an underlying spacetime fault complex.
\end{tcolorbox}
From this property, we conclude that ``untwisted'' detectors, $D_X^{b=0}$, are the same as the $X$-detectors of the underlying CSS circuit.
Additionally, we infer that for all $b$ fulfilling condition C, $D_X^b$ is the kernel of a linear map
\begin{align}
    M^b\colon D_X^{b=0}\to D_X^{b=0},
\end{align}
which we can calculate efficiently.
In other words, $\im(M^b)$ is the space of detectors that are randomized compared to the $b=0$ case.
We can also calculate $\kappa_b$ efficiently.
This gives rise to an efficient algorithm to derive all twisted detectors.

To decide on the decoding success we devise a canonical way to map a correction in the twisted decoding graph onto a correction in the untwisted decoding graph defined by $D_X^{b=0}$, where the logical classes are well-defined.
This canonical mapping is guaranteed by condition C.
This mapping can be seen as creating additional errors in the untwisted decoding graph, which are undetectable in the twisted one.
We call these types of errors \emph{twisted errors}.

\myparagraph{Structure of the simulator}
Based on the formalism sketched above, the rough structure of the resulting algorithm is
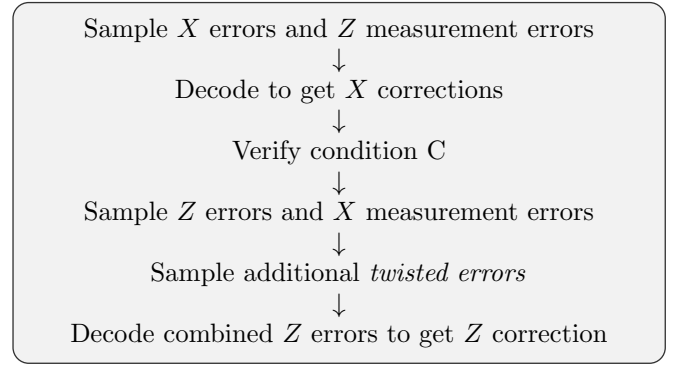
\begin{tcolorbox}[
  colback=white!95!black,   
  colframe=white!20!black, 
  boxrule=0.5pt,
  arc=2mm,
  boxsep=0pt,
  halign=center
]
    Sample $X$ errors and $Z$ measurement errors\\
    $\downarrow$\\
    Decode to get $X$ corrections\\
    $\downarrow$\\
    Verify condition C\\
    $\downarrow$\\
    Sample $Z$ errors and $X$ measurement errors\\
    $\downarrow$\\
    Sample additional \emph{twisted errors}\\
    $\downarrow$\\
    Decode combined $Z$ errors to get $Z$ correction
\end{tcolorbox}

To make a more precise statement about the simulator we need to introduce a precise formalization of third-order circuits.
This is done in Sec.~\ref{sec:algebraic-prop}.
Sec.~\ref{sec:detectors+equivalences} contains general definitions and derivations for third-order circuits.
Sec.~\ref{sec:FT-thirdorder} contains the main results in the paper, with a precise statement of condition C in Def.~\ref{def:condition-C} and the main theorem in Thm.~\ref{thm:twistedconstraints}.
Lem.~\ref{lem:probabilities} guarantees the flatness of the measurement distribution.
In Sec.~\ref{sec:decoding-sim} we present the main simulator in Alg.~\ref{alg:overview} and describe the subroutines.
In Sec.~\ref{sec:interface} we describe how to interface our simulator with Clifford simulators that rely on the usual stabilizer formalism to track the instantaneous Pauli frames.

We review the mathematical background on $i$th order functions and the diagonal gates in the Clifford hierarchy in App.~\ref{app:ithorder} and exemplify the formalism by calculating the set of twisted errors in some simple examples in App.~\ref{app:examples}.

\section{Algebraic description of third-order circuits}\label{sec:algebraic-prop}
A third-order circuit can be represented as a particular type of tensor network~\cite{Bridgeman2017TNtutorial} composed of \emph{signed $Z$ tensors},
\begin{equation}\label{eq:Ztensordef}
  \vcenter{\hbox{%
    \begin{tikzpicture}[scale=1.2, transform shape]
      \node[anchor=north east] at (-0.10,0) {$s$};

      \fill (0,0) circle[radius=4.25pt];

      \foreach \angle/\index in {
        135/i,
         85/j,
         35/{\ldots},
        -15/k
      }{
        \draw[line width=0.6pt]
          (\angle:4.25pt) -- (\angle:0.50cm);

        \node at (\angle:0.66cm)
          {$\scriptstyle \index$};
      }
    \end{tikzpicture}%
  }}
  =
  (-1)^{s i}\delta_{i,j,\ldots,k},
  \qquad
  s\in\{0,1\},
\end{equation}
for an arbitrary number of edges, \emph{signed $X$ tensors},
\begin{equation}\label{eq:Xtensordef}
  \vcenter{\hbox{%
    \begin{tikzpicture}[scale=1.2, transform shape]
      \node[anchor=north east] at (-0.10,0) {$s$};

      \draw[fill=white, line width=1pt] (0,0) circle[radius=4.25pt];

      \foreach \angle/\index in {
        135/i,
         85/j,
         35/{\ldots},
        -15/k
      }{
        \draw[line width=0.6pt]
          (\angle:4.25pt) -- (\angle:0.50cm);

        \node at (\angle:0.66cm)
          {$\scriptstyle \index$};
      }
    \end{tikzpicture}%
  }}
  = \delta_{i+j+...+k=s\mod 2} \qcomma s\in \{0,1\}.
\end{equation}
Additionally, we consider third-order tensors,
\begin{align}\label{eq:third-order-tensordef}
   \vcenter{\hbox{%
    \begin{tikzpicture}[scale=1.2, transform shape]
  \foreach \y in {0.20,0.06,-0.22}{
    \draw[line width=0.6pt]
      (0.30,\y) -- (0.90,\y);
  }
  \node[scale=0.55] at (0.60,-0.045) {$\vdots$};
  \draw[fill=white, line width=1pt]
    (-0.30,-0.30) rectangle (0.30,0.30);
  \node at (0,0) {$S_t$};
  \draw[
    decorate,
    decoration={brace, amplitude=4pt}
  ]
    (1.00,0.26) -- (1.00,-0.27)
    node[midway, right=2pt] {$n$};
\end{tikzpicture}%
  }} = e^{2\pi i S_t[A]}\qcomma S_t\colon \bZ_2^{n} \to \bR/\bZ,
\end{align}
where $n$ is the number of edges of the tensor, and $S_t$ is a third-order function, see App.~\ref{app:ithorder}.

The tensor-network representation of third-order circuits can be obtained directly from the circuit description.
First, each CSS gate can be replaced by a small tensor-network composed of $X$ and $Z$ tensors~\cite{kissinger2022phasefree}, when post-selected onto the $+1$ outcome.
Each $n$-qubit third-level diagonal gate is represented by a single third-order tensor with $n$ open edges.
This tensor is connected to the qubit worldlines, in between the CSS gates on which it acts, via a $Z$ tensor.
In practice, we will consider third-order circuits compiled into a combination of $T$, $CS$ and $CCZ$ gates, whose associated tensors are
\begin{subequations}
\begin{align}
    \vcenter{\hbox{%
   \begin{tikzpicture}[scale=1.2, transform shape]
  \draw[line width=0.6pt]
    (0.30,0) -- (0.80,0);
  \draw[fill=white, line width=1pt]
    (-0.30,-0.30) rectangle (0.30,0.30);
  \node at (0,0) {$T$};
  \node[anchor=west, scale=0.8] at (0.80,0) {$A_1$};
\end{tikzpicture}%
}} =& e^{2\pi i A_1/8}\qcomma\\
\vcenter{\hbox{%
  \begin{tikzpicture}[scale=1.2, transform shape]
    \draw[line width=0.6pt]
      (0.30, 0.14) -- (0.80, 0.14);
    \draw[line width=0.6pt]
      (0.30,-0.14) -- (0.80,-0.14);
    \draw[fill=white, line width=1pt]
      (-0.350,-0.30) rectangle (0.350,0.30);
    \node at (0,0) {$CS$};
    \node[anchor=west, scale=0.8] at (0.80, 0.14) {$A_1$};
    \node[anchor=west, scale=0.8] at (0.80,-0.14) {$A_2$};
  \end{tikzpicture}%
}} =& e^{2\pi i A_1A_2/4}\qcomma\\
\vcenter{\hbox{%
  \begin{tikzpicture}[scale=1.2, transform shape]
    \draw[line width=0.6pt]
      (0.30, 0.250) -- (0.90, 0.250);
    \draw[line width=0.6pt]
      (0.30, 0) -- (0.90, 0);
    \draw[line width=0.6pt]
      (0.30,-0.250) -- (0.90,-0.250);
    \draw[fill=white, line width=1pt]
      (-0.45,-0.350) rectangle (0.45,0.350);
    \node at (0,0) {$CCZ$};
    \node[anchor=west, scale=0.8] at (0.90, 0.30) {$A_1$};
    \node[anchor=west, scale=0.8] at (0.90 , 0) {$A_2$};
    \node[anchor=west, scale=0.8] at (0.90,-0.30) {$A_3$};
  \end{tikzpicture}%
}}=& e^{2\pi i A_1A_2A_3/2}.
\end{align}

This procedure leads to a tensor-network representation of the third-order circuit that is composed of $X$ and $Z$ tensors on the qubit worldlines and additional third-order tensors connected to the worldlines via $Z$ tensors.
As a final step, we split each edge that connects two tensors of the same type using the identity rule,
\begin{align}\label{eq:identity-rules}
    \vcenter{\hbox{%
    \begin{tikzpicture}[scale=1.2, transform shape]
  \draw (0,0) -- (0.8,0);
\end{tikzpicture}%
  }} = \vcenter{\hbox{%
    \begin{tikzpicture}[scale=1.2, transform shape]
  \draw (0,0) --coordinate[midway] (midleg) (0.8,0);
  \fill (midleg) circle[radius=3pt];
\end{tikzpicture}%
  }} = \vcenter{\hbox{%
    \begin{tikzpicture}[scale=1.2, transform shape]
  \draw (0,0) --coordinate[midway] (midleg) (0.8,0);
  \draw[fill=white, line width=1pt] (midleg) circle[radius=3pt];
\end{tikzpicture}%
  }},
\end{align}
\end{subequations}
such that no $X$ tensor is connected directly to another $X$ tensor and similarly for the $Z$ tensors.

\subsection{Third-order circuit as path integral}\label{sec:pathintegral}
After representing the third-order circuit in this way, we can start extracting the algebraic data that specifies the circuit.
First, consider only the interior of the circuit, i.e., do not allow for any quantum in- or output.
This is sufficient for the analysis of detectors and equivalences.
In Sec.~\ref{sec:output-states} we present an exact description of the output state of a third-order circuit.

We start by ignoring the third-order tensors, and only consider the CSS part of the network.
Let $T_Z$ be the set of $Z$ tensors and $T_X$ the set of $X$ tensors.
The sub-network formed by $T_Z$ and $T_X$ is a bipartite graph and is described by a single adjacency matrix
\begin{align}\label{eq:ZX-adjacency}
    d_1\colon \bZ_2^{T_Z} \to \bZ_2^{T_X}.
\end{align}
We also refer to that diagram in short with $[T_Z\stackrel{d_1}{\longrightarrow} T_X]$.

\myparagraph{Identifying measurement outcomes and errors in \eqref{eq:ZX-adjacency}}
Each time a third-order circuit is run, measurement outcomes are recorded.
We represent each measurement outcome with a single bit and collect all measurement outcomes into a bitstring 
\begin{align}
    m = (m_x,m_z) \in \bZ_2^{M_X} \oplus \bZ_2^{M_Z},
\end{align}
where $M_X$ is the set of $X$ measurements in the circuit and $M_Z$ the set of $Z$ measurements in the circuit.
Each $m$ can be represented by a set of signs on the $X$ and $Z$ tensors in the tensor-network.
Formally, this defines maps $\bZ_2^{M_X}\hookrightarrow \bZ_2^{T_Z}$ and $\bZ_2^{M_Z}\hookrightarrow \bZ_2^{T_X}$.

Additionally, we consider the case of stochastic Pauli errors, where in each run a combination of Pauli errors happens while executing the circuit.
The errors are modelled to act in between two consecutive gates.
In each run of the circuit the error-pattern is fixed and can be represented by a bitstring
\begin{align}
    e = (e_x,e_z) \in \bZ_2^{P_X}\oplus \bZ_2^{P_Z},
\end{align}
where $P_{X(Z)}$ is the set of Pauli-$X(Z)$ error locations.
Each error location corresponds to an edge between an $X$ and $Z$ tensor in the tensor-network representation.
Using the fusion rule,
\begin{subequations}\label{eq:fusionrules-tensors}
\begin{align}
    \vcenter{\hbox{%
    \begin{tikzpicture}[scale=1.2, transform shape]
      \node[anchor=north east] at (-0.10,0) {$s$};
      \draw[fill=black, line width=1pt] (0,0) circle[radius=4.25pt];
      \foreach \angle/\index in {
        135/i,
         85/j,
         35/{\ldots},
        -15/k
      }{
        \draw[line width=0.6pt]
          (\angle:4.25pt) -- (\angle:0.70cm);
      }
      \draw[fill=black] (35:0.4cm) circle[radius=3pt];
      \node[font=\scriptsize, anchor=south] at (35:0.4cm) {$1$};
    \end{tikzpicture}%
  }} =& \vcenter{\hbox{%
    \begin{tikzpicture}[scale=1.2, transform shape]
      \node[anchor=north east, font=\small] at (-0.10,0) {$s+1$};
      \draw[fill=black, line width=1pt] (0,0) circle[radius=4.25pt];
      \foreach \angle/\index in {
        135/i,
         85/j,
         35/{\ldots},
        -15/k
      }{
        \draw[line width=0.6pt]
          (\angle:4.25pt) -- (\angle:0.50cm);
      }
    \end{tikzpicture}%
  }}
  \qq{and}\\
      \vcenter{\hbox{%
    \begin{tikzpicture}[scale=1.2, transform shape]
      \node[anchor=north east] at (-0.10,0) {$s$};
      \draw[fill=white, line width=1pt] (0,0) circle[radius=4.25pt];
      \foreach \angle/\index in {
        135/i,
         85/j,
         35/{\ldots},
        -15/k
      }{
        \draw[line width=0.6pt]
          (\angle:4.25pt) -- (\angle:0.70cm);
      }
      \draw[fill=white, line width=0.85pt] (35:0.4cm) circle[radius=3pt];
      \node[font=\scriptsize, anchor=south] at (35:0.4cm) {$1$};
    \end{tikzpicture}%
  }} =& \vcenter{\hbox{%
    \begin{tikzpicture}[scale=1.2, transform shape]
      \node[anchor=north east, font=\small] at (-0.10,0) {$s+1$};
      \draw[fill=white, line width=1pt] (0,0) circle[radius=4.25pt];
      \foreach \angle/\index in {
        135/i,
         85/j,
         35/{\ldots},
        -15/k
      }{
        \draw[line width=0.6pt]
          (\angle:4.25pt) -- (\angle:0.50cm);
      }
    \end{tikzpicture}%
  }}\; ,
\end{align}
\end{subequations}
we can uniquely map any error configuration onto a sign configuration on the $X$ and $Z$ tensors.
Formally, this defines injections $\bZ_2^{P_Z} \hookrightarrow\bZ_2^{T_Z}$ and $\bZ_2^{P_X} \hookrightarrow\bZ_2^{T_X}$.
Together with the measurements, we define entire \emph{chains} of signs in the network,
\begin{align}\label{eq:chargeandflux}
    c = e_z + m_x\in \bZ_2^{T_Z}\qq{and} b = e_x + m_z \in \bZ_2^{T_X}.
\end{align}

\myparagraph{Enriching CSS circuits by $i$th order gates}
Each third-order gate is represented by a third-order tensor that is connected to $Z$ tensors on the qubit worldlines each of which represents an explicit $\bZ_2$ variable that the tensor network sums over.
Each gate acts on a different subset of these degrees of freedom and we can represent the entirety of the third-order gates in the circuit as a third-order function
\begin{align}
    S:\bZ_2^{T_Z} \to \bR/\bZ.
\end{align}
Concretely, we index each $T$ gate in the circuit by a $\bZ_8 = \{0,1,...,7\}$ variable, each $CS$ gate by a $\bZ_4 = \{0,1,2,3\}$ and each $CCZ$ by a $\bZ_2=\{0,1\}$ variable on $T_Z$.
Taken together, the third-order gates in a third-order circuit are given by an array
\begin{equation}
\begin{split}
    s =& (s^{(T)}, s^{(CS)}, s^{(CCZ)})\\
    &\in \bZ_8^{T_Z}\oplus \bZ_4^{\mathrm{pairs}(T_Z)}\oplus\bZ_2^{\mathrm{triples}(T_Z)},
\end{split}
\end{equation}
where $\mathrm{pairs}(T_Z) = \{(t,t')\in T^{\times 2}\;|\; t\neq t'\}$ is the set of all disjoint, unordered, pairs in $T_Z$ and, similarly, $\mathrm{triples}(T_Z) = \{(t,t',t'')\in T^{\times 3}\;|\;t\neq t'\neq t''\neq t\}$ is the set of all disjoint triples in $T_Z$.
We denote a generic entry of the array $s$ by $(s^{(T)}_{a}, s^{(CS)}_{bc}, s^{(CCZ)}_{def})$, with $a,b, ...,f\in T_Z$.
The associated third-order function is given by the polynomial
\begin{widetext}
\begin{align}\label{eq:spacetime-3rd-phasepoly}
    S[A] = \sum_{i\in T_Z} \frac{s^{(T)}_{i}}{8} \overline{A}_i + \sum_{(i,j)\in \mathrm{pairs}(T_Z)} \frac{s^{(CS)}_{(ij)}}{4} \overline{A}_i \overline{A}_j + \sum_{(i,j,k) \in \mathrm{triples}(T_Z)} \frac{s^{(CCZ)}_{(ijk)}}{2} \overline{A}_i \overline{A}_j \overline{A}_k,
\end{align}    
\end{widetext}
where $\overline{A}$ denotes the lift of a $\bZ_2$-vector to a $\bZ$-vector that replaces $0, 1\in \bZ_2$ with the corresponding integers $0,1\in \bZ$.

\begin{remark}
    Eq.~\eqref{eq:spacetime-3rd-phasepoly} can be understood as a phase polynomial representation of the diagonal gates in the circuit, akin to the phase-polynomial representation of a single diagonal gate on a set of qubits, as introduced in Ref.~\cite{Cui2017diagonal} and heavily used in Ref.~\cite{Campbell_2017}.
    In that sense, $S$ is a \emph{spacetime phase polynomial} associated to the third-order circuit.
\end{remark}

We are now in a position to give an explicit expression to evaluate any third-order circuit, defined by an adjacency map $d_1:\bZ_2^{T_Z}\to \bZ_2^{T_X}$, a third-order function $S:\bZ_2^{T_Z} \to \bR/\bZ$, and an arbitrary sign configuration of $Z$ and $X$ tensors, captured by $c\in \bZ_2^{T_Z}$ and $b \in \bZ_2^{T_X}$.
The latter include both errors and measurement outcomes in the circuit, see Eq.~\eqref{eq:chargeandflux}.
We treat both $c$ and $b$ as \emph{external} variables, indexing the family of circuits obtained from all possible error and measurement outcome configurations.
Combining the expressions for the elementary tensors in the third-order circuit, see Eq.~\eqref{eq:third-order-tensordef}, we obtain the expression
\begin{align}\label{eq:third-order-circuit-bulk}
    C_{b,c} \propto \sum_{A\in \bZ_2^{T_Z}\colon d_1A = b} e^{2\pi i S[A]} (-1)^{\langle c, A\rangle},
\end{align}
up to an irrelevant global normalization factor.
$\langle c,A\rangle = \sum_i c_iA_i\mod 2$ denotes the standard inner product on $\bZ_2^{T_Z}$.
As this expression sums over all allowed spacetime configurations of the qubits $A$, weighted by their respective contributions, given by the diagonal gates in the circuit, we call this the \emph{path integral} expression of the circuit~\cite{Bauer2024pathintegral}.

\subsection{Detectors and equivalences}\label{sec:detectors+equivalences}

Decoding active QEC circuits is based on linear constraints on the measurement outcomes, depending on the error.
In the following, we derive the linear constraints from Eq.~\eqref{eq:third-order-circuit-bulk} based on generic properties of $S$ that will be made concrete in Sec.~\ref{sec:FT-thirdorder} .

For the rest of the section we describe how detectors can be characterized using the algebraic data from Sec.~\ref{sec:pathintegral}.
In contrast to the introduction before it is important to now consider third-order circuits with open edges, i.e. $C_{b,c}$ now corresponds to a linear operator from the input space to the output space of the circuit.

\begin{remark}
    Below, we make some definitions and explanations in a language that resembles the algebraic data we are working with.
    For most of the terms, however, equivalent notions have been developed for Clifford circuits.
    For some readers, the term \emph{equivalences} might be familiar as faults that are gauge operators of an auxiliary spacetime subsystem code that captures some QEC properties of the circuits~\cite{bacon2015sparse, delfosse2023spacetime, pesah2025spacetimecode}.
    In Ref.~\cite{blackwell2025codedistance} these types of faults were simply called \emph{trivial} faults and in Refs.~\cite{xyzruby, rüsch2026completeness} \emph{inconsequential faults}.
    We use the term detectors to refer to linear constraints among the measurement outcomes of the circuit.
    Often, these constraints are said to emerge in the absence of noise.
    We make it explicit that linear constraints exist for \emph{each} error pattern.
    In Clifford circuits the number of independent constraints is the same for each error pattern.
    Our work adds to the existing terminology by generalizing it to third-order circuits and highlighting that the presence of $X$ errors decreases the number of independent constraints on the $Z$ errors.
\end{remark}

\myparagraph{Detectors for $X$ errors}
The $X$ errors are constrained by relations among $Z$-symmetries of the circuit.
These take the role of equivalences among $Z$ errors in the circuit.
Since the non-Clifford gates in the circuit are diagonal all Pauli-$Z$ operators can be freely commuted past the gates, or equivalently, past the $Z$-tensors that connect the qubit worldlines to the third-order tensors defining $S$.
As a consequence, $b$ is constrained in the same way as in a CSS circuit.

To derive the $Z$-symmetries and relations among them, consider a single $X$ tensor. It has the following Pauli symmetry
\begin{align}\label{eq:Xtensor-symmetry}
  \vcenter{\hbox{%
    \begin{tikzpicture}[scale=1.2, transform shape]
      \node[anchor=north east] at (-0.10,0) {$s$};
      \draw[fill=white, line width=1pt] (0,0) circle[radius=4.25pt];
      \foreach \angle/\index in {
        135/i,
         85/j,
         35/{\ldots},
        -15/k
      }{
        \draw[line width=0.6pt]
          (\angle:4.25pt) -- (\angle:0.50cm);
      }
    \end{tikzpicture}%
  }}
  = (-1)^s\:
\vcenter{\hbox{%
  \begin{tikzpicture}[scale=1.2, transform shape]
    \node[anchor=north east] at (-0.10,0) {$s$};
    \foreach \angle [count=\leg] in {135,85,35,-15}{
      \draw[line width=0.6pt]
        (\angle:4.25pt)
        --
        coordinate[midway] (mid-\leg)
        (\angle:0.75cm);
      \fill (mid-\leg) circle[radius=2.75pt];
      \path
        (mid-\leg)
        ++({\angle-90}:0.2cm)
        node[font=\scriptsize, inner sep=0pt] {$1$};
    }
    \draw[fill=white, line width=1pt]
      (0,0) circle[radius=4.25pt];
  \end{tikzpicture}%
}}
\end{align}
for any number of outgoing edges.
By construction, every $X$ tensor is connected to a $Z$ tensor.
This allows us to push any Pauli-$Z$ operator from Eq.~\eqref{eq:Xtensor-symmetry} onto the adjacent $Z$ tensors by flipping their sign.
This $Z$-like symmetry can be expressed as
\begin{align}\label{eq:circuit-symmetry-boundary}
    C_{b,c} = (-1)^{\langle b, f\rangle} C_{b,c+d_1^T f}
\end{align}
for every $f\in \bZ_2^{T_X}$.
Note that the phase factor must be $+1$ whenever $d_1^T f = 0$.
This leads to a linear constraint of the form
\begin{align}\label{eq:b-cocyclecondition}
    \langle b, d_Z\rangle \stackrel{!}{=} 0\quad\forall d_Z\in \ker(d_1^T).
\end{align}
The elements in $\ker(d_1^T)$ can be understood as relations among the $Z$-like symmetries of the individual tensors that arise from the specific connectivity among the tensors in the network, defined by $d_1$.
The image of $d_1^T$ does exactly correspond to the $Z$-like \emph{inconsequential faults}, as introduced in Ref.~\cite{xyzruby}.
Operationally, these are combinations of Pauli-$Z$ errors and $X$ measurement outcomes that do not affect the action of the underlying circuit, as the associated tensor networks agree up to a global phase.

\begin{definition}\label{def:ZequivalencesXsyndrome}
For each $f\in \bZ_2^{T_X}$ we define an \emph{$Z$-equivalence} $g_f = d_1^T f\in \bZ_2^{T_Z}$.
We say each $d_Z\in \ker(d_1^T)$ defines a \emph{$X$-detector} via a linear constraint $\langle b, d_Z\rangle = 0$.
Let $D_Z = \{d_i\}_i$ be a generating set for $\ker(d_1^T)$.
We define the \emph{$X$-syndrome map}
\begin{align}
\begin{split}
    d_2\colon \bZ_2^{T_X} &\to \bZ_2^{D_Z};\\
    b &\mapsto \mqty(\vdots \\ \langle b, d_i\rangle \\ \vdots),
\end{split}    
\end{align}
where $\bZ_2^{D_Z}$ is the (abstract) space formed by all $X$ detectors.
For a given vector of $Z$-measurement outcomes $m_z\in \bZ_2^{T_X}$ we define the associated \emph{syndrome} $s_z = d_2m_z$.
\end{definition}

Indeed, this notion of a detector coincides with other definitions that have been established for Clifford circuits and defines a corresponding classical (linear) decoding problem.

\begin{lemma}\label{lem:Xdetectorswork}
    Let $m_z \in \bZ_2^{T_X}$ be a vector of $z$-measurement outcomes with syndrome $s_z$ and $e_x\in\bZ_2^{T_X}$ a vector of $X$ errors.
    There is a linear relation between the $X$ errors and the syndrome,
    \begin{align}\label{eq:X-decoding}
        d_2 e_x = s_z.
    \end{align}
\end{lemma}

\begin{proof}
A short calculation, using the definition of $b$ from Eq.~\eqref{eq:chargeandflux}, verifies that $d_2b \stackrel{\eqref{eq:b-cocyclecondition}}{=} 0\; \Leftrightarrow\; d_2 e_x = d_2 m_z$.
\end{proof}

A decoder has to solve Eq.~\eqref{eq:X-decoding} for $e_x$. However, it is defined up to inconsequential errors leading to a high-dimensional solution space.
This space is split into logical sectors, which we define in Sec.~\ref{sec:output-states} when introducing quantum input- and output boundaries to the circuit.
These will be important to define the decoding success/failure.

\myparagraph{Detectors for $Z$ errors}
Similar to $X$ errors, $Z$ errors are constrained by relations among $X$ symmetries of the circuit which can be understood as equivalences among the $X$ errors.
These are generated by a subset of $X$-symmetries of the $Z$ tensors in the network,
\begin{align}\label{eq:Ztensor-symmetry}
  \vcenter{\hbox{%
    \begin{tikzpicture}[scale=1.2, transform shape]
      \node[anchor=north east] at (-0.10,0) {$s$};
      \draw[fill=black, line width=1pt] (0,0) circle[radius=4.25pt];
      \foreach \angle/\index in {
        135/i,
         85/j,
         35/{\ldots},
        -15/k
      }{
        \draw[line width=0.6pt]
          (\angle:4.25pt) -- (\angle:0.50cm);
      }
    \end{tikzpicture}%
  }}
  = (-1)^s\:
\vcenter{\hbox{%
  \begin{tikzpicture}[scale=1.2, transform shape]
    \node[anchor=north east] at (-0.10,0) {$s$};
    \foreach \angle [count=\leg] in {135,85,35,-15}{
      \draw[line width=0.6pt]
        (\angle:4.25pt)
        --
        coordinate[midway] (mid-\leg)
        (\angle:0.75cm);
      \draw[fill=white, line width=0.7pt] (mid-\leg) circle[radius=2.75pt];
      \path
        (mid-\leg)
        ++({\angle-90}:0.2cm)
        node[font=\scriptsize, inner sep=0pt] {$1$};
    }
    \draw[fill=black, line width=1pt]
      (0,0) circle[radius=4.25pt];
  \end{tikzpicture}%
}},
\end{align}
for any number of edges.
One might expect that the $X$ tensors on the edges above can be merged into the adjacent tensors to recover an equation similar to Eq.~\eqref{eq:Xtensor-symmetry}.
The third-order tensors, however, do constrain the mobility of $X$ tensors through the circuit which removes the $X$ symmetries in the local neighborhood of non-trivial third-order tensors.
Explicitly, for arbitrary $t\in \bZ_2^{T_Z}$ we find that 
\begin{subequations}
\begin{align}\label{eq:firing-Xs-third-order}
    C_{b,c} =& \sum_{A\in \bZ_2^{T_Z}\colon d_1A = b} e^{2\pi i S[A]} (-1)^{\langle c, A\rangle}\\
    \stackrel{\eqref{eq:Ztensor-symmetry}}{=}& (-1)^{\langle c,t\rangle}\sum_{A\in \bZ_2^{T_Z}\colon d_1 A = b+ d_1 t} e^{2\pi i S[A+t]} (-1)^{\langle c, A\rangle}.
\end{align}
\end{subequations}
The right hand side in the equation above corresponds to an effective transformation on data of $C_{b,c}$,
\begin{align}
    b \mapsto b + d_1 t\qq{and} S[A] \mapsto S[A+t],
\end{align}
which only defines an $X$-like equivalence $b\sim b' = b+d_1t$ if $S[A+t] = S[A]$.
This holds trivially for $t$ that is supported only on $Z$-tensors that are not connected to any third-order tensor.
Using the equations above we can exactly characterize the subset of $t\in \bZ_2^{T_Z}$ that share support with $S$ and define an $X$ equivalence.
For an $X$ equivalence, defined via $t\in \bZ_2^{T_Z}$ we find that
\begin{align}
    C_{b,c} = (-1)^{\langle c, t\rangle} C_{b+d_1t, c}.
\end{align}
This leads to one non-trivial relation among these symmetries for each $t\in \ker(d_1)$, from which we obtain linear relations of the form
\begin{align}
    \langle c, d_X\rangle \stackrel{!}{=} 0\;\quad \forall d_X\in \ker(d_1).
\end{align}
More generally, if $S[A]-S[A+t]\in \{0,1/2\}$ only depends on $b$ we obtain an affine linear constraint on $c$.

We summarize the quantities we describe above in the following definition.

\begin{definition}\label{def:thirdorder-symmetries}
    Let $C_{b,c}$ be a third-order circuit defined by a third-order function $S\colon \bZ_2^{T_Z}\to \bR/\bZ$ and fixed $b,c$ and recall that $d_1A=b$ can be considered (see Eq.~\eqref{eq:third-order-circuit-bulk}).
    We first define the set of \emph{$S$-symmetries} for each $b$,
    \begin{equation}
    \begin{split}
        \Sym_S^b = \bigg\{f\in \bZ_2^{T_Z}\;|\; S(A+f)-S(A) =\frac{\kappa_b(f)}{2}& \\
        \forall A \colon d_1A = b& \bigg\},
    \end{split}
    \end{equation}
    where $\kappa_b\colon \bZ_2^{\Sym_S^b}\to \bZ_2$ only depends on $b=d_1A$.
    $2\bullet$ denotes the map that identifies integer multiples of $\{0, 1/2\}^\bullet$ with $\{0, 1\}^\bullet$, the canonical representatives of elements in $\bZ_2^\bullet$.
    We denote the set of \emph{twisted $X$-equivalences} as
    \begin{align}
        E_X^b = d_1\,\Sym_S^b = \{d_1 f \in \bZ_2^{T_X}\;|\; f\in \Sym_S^b\}.
    \end{align}
    We say that each $d_X\in \ker(d_1)\cap \Sym_S^b$ defines a \emph{twisted $Z$-detector} via an (affine) linear constraint $\langle c, d_X\rangle = \kappa_b(d_X)$.
    Let $D_X^b = \{d_i\}_i$ be a generating set for $\ker(d_1)\cap \Sym_S^b$. We define the \emph{twisted $Z$-syndrome map}
    \begin{align}
    \begin{split}
        \partial_0^b\colon \bZ_2^{T_Z} &\to \bZ_2^{D_X^b};\\
        c &\mapsto \mqty(\vdots\\ \langle c, d_i\rangle \\ \vdots),
    \end{split}
    \end{align}
    where $\bZ_2^{D_X^b}$ is the (abstract) space spanned by all $Z$ detectors.
    For a given vector of $X$-measurement outcomes $m_x\in \bZ_2^{T_Z}$ we define the associated \emph{syndrome} $s_x = \partial_0^b m_x$.
    For later convenience, we also introduce its transpose map explicitly as $d_0^b=(\partial_0^b)^T$.
\end{definition}

\begin{remark}
    We add the word ``twisted'' in the paragraphs and definition above to highlight that the spaces depend on $b$ explicitly. In particular this implies that the space of $Z$ detectors and equivalences depends on the $X$ errors and $Z$-measurement outcomes explicitly. 
    This is in line with the conclusion in Ref.~\cite{bombin2018propagation} which considered transversal third-order gates acting on a state with residual errors and found that not all stabilizers are preserved by the logic gate.
    We want to highlight, however, that we do not require the transversal implementation of the third-order gate but can capture any third-order circuit and extract its QEC properties.
\end{remark}

\begin{corollary}
    $\bZ_2^{D_X^b}$ is a subspace of $\bZ_2^{T_Z}$ and $E_X^b$ is a subspace of $\bZ_2^{T_X}$.
\end{corollary}
\begin{proof}
    Need to check closure of $D_X^b$ and $E_X^b$ under addition in $\bZ_2^{T_Z}$ and $\bZ_2^{T_X}$, respectively. The former follows from the definition and the latter statement then from linearity of $d_1$.
\end{proof}

\begin{lemma}\label{lem:Zdetectorswork}
    Let $m_x \in \bZ_2^{T_Z}$ be a vector of $x$-measurement outcomes with syndrome $s_x$ and $e_z\in\bZ_2^{T_Z}$ a vector of $Z$ errors.
    There is a linear relation between the $Z$ errors and the syndrome,
    \begin{align}\label{eq:Z-twisted-constraints}
        \partial_0^b e_z = s_x + \kappa_b,
    \end{align}
    where $\kappa_b$ is interpreted as a vector in $\bZ_2^{D_X^b}$.
\end{lemma}
\begin{proof}
    By definition, the proof proceeds analogously to the proof of Lem.~\ref{lem:Xdetectorswork}.
\end{proof}

Often, a circuit with $b\neq 0$ can detect strictly fewer $Z$ errors, i.e. $\ker(\partial_0)\subsetneq \ker(\partial_0^b)$.
In that case, we want to give the new set of undetectable errors a name.

\begin{definition}[twisted $Z$ errors]\label{def:twistedZerrors}
    Assume $\ker(\partial_0^{b=0}) \subseteq \ker(\partial_0^b)$.
    Consider a decomposition $\ker(\partial_0^b) =\ker(\partial_0^{b=0}) \oplus \mathcal{E}_Z^b $.
    Any $\mathcal{E}_Z^b$, defined by a set of generators $\{c_{\mathrm{tw}}\}_{c_{\mathrm{tw}}}$, is called the space of \emph{twisted errors}.
    Two $\mathcal{E}_Z^b$ are considered equivalent if their respective generators $c_{\mathrm{tw}}$ differ by elements in $\ker(\partial_0^{b=0})$.
\end{definition}

\begin{remark}
Lem.~\ref{lem:Zdetectorswork} suggests that when running the third-order circuit with measurement outcomes $m_x$, a decoder has to solve Eq.~\eqref{eq:Z-twisted-constraints} for $e_z$.
However, this is not realistically possible as the linear constraints themselves depend on $b$, a quantity that includes unknown Pauli-$X$ errors. In particular, we cannot assume any precise knowledge about the location of Pauli-$X$ errors in the circuit.
Hence, any decoder must make an assumption on the precise structure of $b$.
A naive decoder, for example, would just ignore the dependence on $b$ and assume that the previous decoding step identified the errors perfectly, e.g. $b=0$.
A more sophisticated method involves partial knowledge about $b$, since it also includes the corrections that were done.
Moreover, it can make a guess on the error $e_x$, and decode based on a ``guessed'' decoding graph characterized by another parity-check matrix $\tilde{\partial}_0^b$.
This approach was used in Ref.~\cite{highthresholdJIT} and significantly improved the threshold.
\end{remark}

\subsection{Fault-tolerant third-order circuits}\label{sec:FT-thirdorder}
In this section we show that in the case where $b=0$, our framework provides an algebraic characterization of a \emph{fault-tolerant} insertions of non-Clifford gates, in terms of a cohomology invariant of an underlying 4-term chain complex.
We show that this also guarantees that localized errors in the circuit cannot cause any undetectable logical error and derive our main results that we use for our simulator in Sec.~\ref{sec:decoding-sim}.

To start, consider a CSS circuit, in which $S[A]\equiv 0$.
In that case $D_X^b = \ker(d_1) =: D_X$, independent of $b$.
We now collect the algebraic data above into a single mathematical object. This will allow us to identify special classes of functions $S$ for which the calculation of twisted detectors and equivalences is tractable for certain $b$.
In practice, the circuit must include adaptive corrections, decided on by a decoder, to guarantee that $b$ is localized sufficiently to preserve the fault-distance scaling of the circuit and to allow for efficient simulability of the measurement and error statistics.

The spaces and maps between them that were introduced in the previous two sections can be combined to a sequence consisting of 4 terms
\begin{align}\label{eq:faultcomplex-full}
    \bZ_2^{D_X} \stackrel{d_0}{\longrightarrow} \bZ_2^{T_Z} \stackrel{d_1}{\longrightarrow} \bZ_2^{T_X} \stackrel{d_2}{\longrightarrow} \bZ_2^{D_Z}.
\end{align}
It holds that $d_{i+1}\circ d_{i} = 0\qcomma i=0,1$ which makes the above sequence a (co)chain complex where we denote the elements in the left-most space as 0-(co)chains and the other ones in increasing order from left to right.
This is exactly the \emph{fault-complex} of the underlying CSS circuit~\cite{Hillmann2025singleshot}.

We consider the ensemble of (post-selected) circuits $\{C_{b,c}\}_{b,c}$.
In each run, a different post-selection is observed, which we call a \emph{trajectory}.
In the following, we consider families of third-order circuits defined by families of bi-partite CSS tensor networks $[T_Z\stackrel{d_1}{\to}T_X]$ and third-order functions $S$.
The family is indexed by the \emph{spacetime volume} $\mathrm{Vol} = \abs{T_Z}\abs{T_X}$.

For the rest of this paper we consider a \emph{fault tolerant LDPC} family of circuits, a family of circuits such that the associated complex has sparse coboundary operators and increasing
\begin{equation*}
\begin{split}
    \qq{1-systole} d_Z =& \min(\abs{c}\;|\; c\in \ker(d_0^T)\backslash \im(d_1^T))\\
    \qq{and 2-cosystole} d_X =& \min(\abs{b}\;|\; b\in \ker(d_2)\backslash \im(d_1)),
\end{split}
\end{equation*}
where $\abs{\bullet}$ denotes the $l_1$-norm of a binary vector.
As derived above, these correspond to the distances of the pure $Z$ and $X$ decoding problems.
The actual fault-distance of the entire circuit is ill-defined in general and only upper bounded by $\min(d_X,d_Z)$, as non-trivial $X$ errors in the circuit lead to a reduction of the fault distance against $Z$ errors, due to twisted errors.

To obtain the correct measurement distribution of $X$ measurements, the twisted detectors have to be calculated for the $b$ in question.
A special class of functions for which this is tractable for the important sets of $b$, for which the $X$ decoding succeeds and the fault-distance does not drop to 1, is the class of \emph{local cohomology invariants} of the chain complex in Eq.~\eqref{eq:faultcomplex-full}, where the third-order gates have sparse support. We do not demand geometric locality.

\begin{definition}[local cohomology invariant]\label{def:local-invariant}
A function $S:\bZ_2^{T_Z}\to \bR/\bZ$ for which
\begin{align}
    S[A]  = S[A+d_0v] \quad \forall A\colon d_1A = 0,
\end{align}
is a \emph{cohomology invariant}.
A function $S:\bZ_2^{T_Z}\to \bR/\bZ$ that can be written in the form
\begin{align}
    S = \sum_{\ell\in \mathcal{P}(T_Z)} S_\ell,
\end{align}
where $\mathcal{P}(T_Z)$ denotes the powerset of $T_Z$, with
\begin{itemize}
    \item $\sup(S_{\ell}) = \ell$,
    \item \emph{(finite range)} There exists a global constant $r\in \bZ^+$ such that $S_{\ell} = 0\;\forall \ell\colon \abs{\ell}>r$, and
    \item \emph{(bounded gate density)} There exists a global constant $\rho\in \bZ^+$ such that for any $t\in T_Z$, $\abs{\{S_{\ell}\neq 0\;|\; t\in \ell\}}<\rho$
\end{itemize}
is called \emph{local}.
\end{definition}

\begin{corollary}\label{cor:SzeroCoboundaries}
    A third-order cohomology invariant vanishes on coboundaries, $S[d_0 v] = 0\;\forall v\in \bZ_2^{D_X}$.
\end{corollary}
\begin{proof}
    Follows directly from the general form of Eq.~\eqref{eq:spacetime-3rd-phasepoly} and that $S[d_0 v] = S[0]$.
\end{proof}

The main reason we consider third-order circuits defined by cohomology invariants is that the space of twisted detectors is related to the space of detectors in the underlying CSS circuit (``untwisted detectors'').
We start with stating the simple case, when $b=0$.

\begin{lemma}
    Let $C_{0,c}$ be a trajectory of a third-order circuit with a third-order cohomology invariant, and $b=0$. In that case, the space of detectors is independent of $S$. It holds that $D_X^{b=0} = D_X$, from Eq.~\eqref{eq:faultcomplex-full}.
\end{lemma}
\begin{proof}
    Follows directly from the fact that $S$ is a cohomology invariant, since $d_1A = 0$.
\end{proof}

In the following, we present when and how the space of twisted detectors can be computed for non-trivial $b$.
We consider the setting in which all gates and constraints in the circuit are local.

We say a family of third-order circuits is \emph{sparse} or \emph{LDPC} if all the coboundary operators in the associated fault complex (see Eq.~\eqref{eq:faultcomplex-full}) have row- and column weight of $O(1)$ and $S$ is a local cohomology invariant.

In LDPC circuits we can bound the impact of a non-trivial 2-cocycle $b$ on the distribution of 1-chains, the $Z$-like errors.
To that end, we consider a trivial 2-cocycle, $b=d_1A \in\ker(d_2)$ for some \emph{filling} 1-chain $A$ and define condition C, a sufficient condition that the third-order circuit is fully characterized by twisted constraints.

\begin{definition}[Condition C]\label{def:condition-C}
Let $C_{b,c}$ be a third-order circuit with $b\in\im(d_1)$.
We say it fulfils \emph{condition C} if there exists $A\in \bZ_2^{T_Z}$ such that $d_1A=b$ and the 1-cycles that are contained in the union of all subsets $\ell\subseteq T_Z$ from Def.~\ref{def:local-invariant} that intersect with the support of $A$ only contain trivial 1-cycles, elements in $\im(d_0)$.
\end{definition}

\begin{remark}
    Def.~\ref{def:condition-C} can be interpreted as saying that the code obtained from puncturing the quantum code obtained around $\bZ_2^{T_Z}$ from puncturing onto the support of $A$ in the $X$ basis does not contain any non-trivial logical $Z$ operators.
    In this case, the punctured code can be obtained by measuring out the complement of $\sup(A)$ in the $X$ basis.
\end{remark}

\begin{remark}
    The set of fillings for a fixed $b$ can be sorted into relative cohomology classes and they are in general non-canonically in bijection with the cohomology classes for $b=0$.
    Condition C can be viewed as a canonical way to define the trivial class in the relative cohomology classes.
\end{remark}

\begin{definition}[$M$-matrices]\label{def:Mmatrices}
    Let $C_{b,c}$ be a third-order circuit fulfilling condition C.
    We define the matrix $M^b\colon \bZ_2^{D_X} \to \bZ_2^{D_X}$ by
    \begin{align}
        M^b_{w, v} = \langle w, M^b v\rangle = 2(\Delta^{(3)}S)(A,d_0v, d_0w),
    \end{align}
    for $A\colon d_1A = b\qcomma v,w\in \bZ_2^{D_X}$, where $2\bullet$ denotes the map that identifies integer multiples of $\{0, 1/2\}$ with $\{0, 1\}\simeq\bZ_2$, and a related map $\widetilde{M}^b\colon \bZ_2^{D_X} \to \bZ_2^{T_Z}$ via
    \begin{align}
        \widetilde{M}^b_{e,v} = \langle e, \widetilde{M}^bv\rangle = 2(\Delta^{(3)}S)(A_0, e, d_0v),
    \end{align}
    for $v\in \bZ_2^{D_X},e\in \bZ_2^{T_Z}$.
    We call these two matrices \emph{$M$-matrices}.
\end{definition}

We use the $M$-matrices to characterize the twisted constraints and errors.
The main properties of twisted constraints are summarized in the next theorem.

\begin{theorem}\label{thm:twistedconstraints}
    Let $C_{b,c}$ be a third-order circuit which fulfills condition C with a filling $A_0\colon d_1A_0=b$.
    Consider the $M$-matrices as defined in Def.~\ref{def:Mmatrices}.
    Then,
    \begin{itemize}
        \item The space of twisted $Z$ detectors is given by
        \begin{align}
            \bZ_2^{D_X^b} = \ker(M^b),
        \end{align}
        \item The twisted constraint matrix is given by
        \begin{align}
            \partial_0^b = K^b d_0^T,    
        \end{align}
        where $(K^b)^T\colon \bZ_2^{D_X^b}\to \bZ_2^{D_X}$ is an injective map.
        \item A valid choice of twisted $Z$ errors is
        \begin{align}
            \mathcal{E}_Z^b = \im(\widetilde{M}^b),
        \end{align}
        i.e. an i.i.d random twisted error together with a deterministic error $c_0$, such that $\partial_0^b c_0=\kappa_b$, reproduces the syndrome distribution.
        \item The entries of expected twisted detector values are
        \begin{align}
            (\kappa_b)_{i} = 2(\Delta^{(2)}S)(A_0, d_0(K^b)^T\vb{e}_i),    
        \end{align}
        where $\vb{e}_i$ denotes the basis vector labelled with $i\in D_X^b$.
        It is independent of the chosen contractible representative for $A_0$ as long as it fulfils condition C.
    \end{itemize}
\end{theorem}

This theorem summarizes how to extract all the necessary data to simulate the decoding problem of a third-order circuit.
It relies on first simulating the $X$ decoding and, upon success, the correct distribution of $Z$ errors can be sampled, by computing the $M$-matrices.

In order to prove that theorem we need a small corollary.

\begin{corollary}\label{cor:prop-Mb}
    The matrices in Def.~\ref{def:Mmatrices} have the following properties:
    \begin{itemize}
        \item $M^b$ is symmetric, i.e. $M^b_{v,w} =M^b_{w, v}$.
        \item $M^b$ is independent of the chosen representative $A\in \{A\;|\; d_1A = b\}$.
\end{itemize}
\end{corollary}

\begin{proof}
$M^b$ is symmetric by Cor.~\ref{cor:derivates-symmetric}, as it is defined by a derivative.

To show that $M^b$ is independent of the chosen $A$, we consider the difference of matrix elements obtained from two different $A\neq A'$, with the same coboundary $d_1A = b$, using tri-linearity of $\Delta^{(3)}S$,
\begin{equation}
\begin{split}
    (\Delta^{(3)}S)(A,d_0v, d_0w) - (\Delta^{(3)}S)(A',d_0v, d_0w)\\
    = (\Delta^{(3)}S)(A-A',d_0v, d_0w).
\end{split}
\end{equation}
Expanding the derivative and using cohomology invariance of $S$ and Cor.~\ref{cor:SzeroCoboundaries} shows that this expression evaluates to 0, as the latter two arguments are fixed to coboundaries.
\end{proof}

The correctness of the sampled 1-chains in the simulator is guaranteed by the following lemma.

\begin{lemma}\label{lem:probabilities}
    Let $C_{b,c}$ be a third-order circuit where $b$ fulfills condition C, with filling $A_0$.
    Then, any $c$ that fulfils the twisted constraints in Thm.~\ref{thm:twistedconstraints} appears with the same probability.
\end{lemma}

\begin{proof}
    To show that, for fixed $b$, the probability distribution for twisted errors is flat within a subspace spanned by $\Tilde{M}^b$, with twisted detector values of $\kappa_b$, we express the linear operator from input- to output space, for fixed $b,c$ analogously to Eq.\,\ref{eq:third-order-circuit-bulk},
\begin{equation}
\begin{split}
    \langle a_{\mathrm{in}}|C_{b,c}|a_{\mathrm{out}}\rangle =& \sum_{A\colon d_1A=b} \delta_{A|_{\mathrm{in}}, a_{\mathrm{in}}} \delta_{A|_{\mathrm{out}}, a_{\mathrm{out}}}\\
    &\times e^{2\pi i S[A]}(-1)^{\langle c, A\rangle},
\end{split}
\end{equation}
where $a_{\mathrm{in}}$ and $a_{\mathrm{out}}$ denote computational basis states on the input and output space respectively.
This expression holds up to a global normalization constant that does not depend on $b,c$, or $a_{\mathrm{in}}$, $a_{\mathrm{out}}$.
We rewrite the sum over relative cohomology classes $[A_0] \in (\ker(d_1)+A')/\im(d_0)$, where $d_1A'=b$,
\begin{widetext}
\begin{align}
    \langle a_{\mathrm{in}}|C_{b,c}|a_{\mathrm{out}}\rangle = \sum_{[A_0]} e^{2\pi i S[A_0]}(-1)^{\langle c, A_0\rangle}\sum_{v} \delta_{A+A_0|_{\mathrm{in}}, a_{\mathrm{in}}} \delta_{A+A_0|_{\mathrm{out}}, a_{\mathrm{out}}} e^{2\pi i (S[A_0+d_0v] - S[A_0])}(-1)^{\langle d_0^Tc, v\rangle}.
\end{align}
\end{widetext}
The $\delta$s pick out a single cohomology class.
The probability of a given $b,c$ in the interior of the circuit is obtained from the absolute value squared of the expression in which the first factor cancels.
Note that $e^{2\pi i (S[A_0+d_0v] - S[A_0])} = e^{-2\pi i (\Delta^{(2)}S)(A_0, d_0v)}$, which shows that the exponent is a second-order function. In fact, since $S[0]=S[dv]=0$, it is a quadratic function in $v$, for fixed, contractible $A_0$.
After tracing out the inside and outside basis states we are left with an expression for the probabilities
\begin{equation}
\begin{split}
    \Prob_b(c) \propto& \sum_{v,w} e^{-2\pi i (\Delta^{(2)}S)(A_0, d_0v)- (\Delta^{(2)}S)(A_0,w))}\\
    &\times (-1)^{\langle d_0^Tc, v+w\rangle} \\
    =& \sum_v e^{-2\pi i (\Delta^{(2)}S)(A_0, d_0v) } (-1)^{\langle d_0^Tc, v\rangle }\\
    &\times\sum_w (-1)^{\langle w, M^bv\rangle},
\end{split}
\end{equation}
where we changed one summation index $v\mapsto v+w$ and plugged in the definition of $M^b$, Def.~\ref{def:Mmatrices}.
The inner sum evaluates to
\begin{equation}
\begin{split}
    \sum_w (-1)^{\langle w, M^bv\rangle} =& \begin{cases}
        2^{\abs{D_X}} & v\in \ker(M^b)\\
        0 & \mathrm{else}.
    \end{cases}\\
    =& 2^{\abs{D_X}} \delta_{v\in \ker(M^b)}
\end{split}
\end{equation}
and plugging in $\kappa_b(v)/2 = (\Delta^{(2)}S)(A_0, d_0 v)$ for $v\in \ker(M^b)$, we obtain probabilities
\begin{equation}
\begin{split}
    \Prob_b(c) \propto& 2^{\abs{D_X}} \sum_v  \delta_{v\in \ker(M^b)} (-1)^{\langle d_0^Tc, v\rangle +\kappa_b(v)}\\
    =& 2^{\abs{D_X}} 2^{\abs{D_X^b}}\delta_{\langle d_0^Tc, i\rangle = (\kappa_b)_i\;\forall i\in D_X^b},
\end{split}
\end{equation}
which are the same for every syndrome fulfilling the twisted constraints $K^bd_0^T c = \kappa_b$.
And hence, the probability distribution can be modelled by a flat distribution over twisted errors, $\im(\widetilde{M}^b)$.
This calculation also proves the expression for $\kappa_b$.
Note that condition C guarantees that $(\Delta^{(2)}S)(A_0,dv)$ is independent of the chosen representative $A_0$.
\end{proof}

We can now provide the complete proof of our main theorem.

\begin{proof}(of Thm.~\ref{thm:twistedconstraints})
We prove each bullet point separately.

We can compute explicitly,
\begin{align*}
    M^b_{v,w} =& 2(\Delta^{(3)}S)(A, dv, dw)\\
    =& 2\left((\Delta^{(2)}S)(A, dv) - (\Delta^{(2)}S)(A + dw, dv)\right),
\end{align*}
using that $S$ vanishes on coboundaries, Cor.~\ref{cor:SzeroCoboundaries}.
The twisted detectors are exactly the $v$ for which $dv\in \Sym_S^b$.
These are defined by $(\Delta^{(2)}S)(A, dv)$, and $(\Delta^{(2)}S)(A, dw)$, respectively that only depend on $A$ via $b$.
Shifting $A\mapsto A+dw$ leaves $b$ invariant, so that we obtain
\begin{align}
    M^b_{v,w} = 0
\end{align}
for $dv,dw\in \Sym_S^b$.

$K^b$ is the map that restricts a syndrome on $\bZ_2^{D_X}$ to the syndrome on its subspace spanned by the twisted detectors. $\partial_0^b = K^bd_0^T$ follows directly.

Twisted errors are exactly $Z$ errors $e\in \bZ_2^{T_Z}$ for which $d_0^Te \neq 0$ but $K^bd_0^Te=0$, i.e. they are determined by their syndrome, which is $\ker(K^b)$.
$(K^b)^T$ is a generating matrix of $\ker(M^b)$.
A direct calculation shows that
\begin{align*}
    \ker(K^b) =& \im((K^b)^T)^{\perp}
    = \ker(M^b)^{\perp}\\
    =& \im((M^b)^T) = \im(M^b),
\end{align*}    
using that $(M^b)^T=M^b$ (Lem.~\ref{cor:prop-Mb}).
Since $A_0$ is a filling that fulfills condition C, $\widetilde{M}^b$ is supported in a subspace of $\bZ_2^{T_Z}$ that does not contain any non-trivial 1-cycle.
To show that $d_0^T\widetilde{M}^b = M^b$, we can consider the matrix elements of the left hand side explicitly,
\begin{align}
    \langle w, d_0^T\widetilde{M}^b v\rangle = \langle d_0 w,\widetilde{M}^b v\rangle = \langle w, M^b v\rangle,
\end{align}
by definition of the $M$-matrices, Def.~\ref{def:Mmatrices}.

The value for $\kappa_b$ and its independence of the exact form of $A_0$ follow from the proof of Lem.~\ref{lem:probabilities}.
\end{proof}

\subsection{Logical classes and output states}\label{sec:output-states}

The logical action of third-order circuits is described by non-trivial homology and cohomology classes of the complex from Eq.~\eqref{eq:faultcomplex-full} which only arise for circuits with open input and/or output edges.
If the open edges are connected to $X$ tensors, we use the identity rule, Eq.~\ref{eq:identity-rules}, to add a two-legged $Z$ tensor there. The corresponding value of $A$ corresponds to the computational basis state of that input/output qubit.

In this formulation, the 1-cycles correspond to undetectable combinations of $X$ measurement outcomes and $Z$ errors and 2-cocycles undetectable combinations of $Z$ measurements.
The (co)homology classes of these errors correspond to inequivalent logical errors and provide a precise algebraic description of the $Z$ and $X$ fault distance as 1-systoles and 2-cosystoles, i.e. the minimum-weight representatives in a non-trivial class.
For third-order circuits that implement a logical measurement one must only consider the locally generated part of the boundary maps. Then, the logical outcome corresponds to a non-locally generated class that has a corresponding logical error cochain in the cochain complex.
The same applies to non-locally generated elements in the coboundary maps and corresponding non-local chains in its dual chain complex.

We now use this language to make more operational statements about third-order circuits with errors.
The matrix elements of the linear map implemented by $C_{b,c}$ from an input to an output space can be expressed as
\begin{equation}\label{eq:third-order-circuit-matrixelements}
\begin{split}
    \langle a_{\mathrm{in}}|C_{b,c}|a_{\mathrm{out}}\rangle = \sum_{A\colon d_1A=b} & \delta_{A|_{\mathrm{in}}, a_{\mathrm{in}}} \delta_{A|_{\mathrm{out}}, a_{\mathrm{out}}} \\
    &\times e^{2\pi i S[A]}(-1)^{\langle c, A\rangle},
\end{split}
\end{equation}
where $a_{\mathrm{in}}$ and $a_{\mathrm{out}}$ denote computational basis states on the input and output space respectively.
This is analogous to Eq.~\eqref{eq:third-order-circuit-bulk}. 
Note that this expression is exact for any $b$.
If condition C is fulfilled it is helpful to think of the \emph{twisted} chain complex,
\begin{align}
    \bZ_2^{D_X^b} \stackrel{d_0^b}{\longrightarrow} \bZ_2^{T_Z} \stackrel{d_1^b}{\longrightarrow} \bZ_2^{T_X}\oplus \mathcal{E}_Z^b \stackrel{(d_2\, 0)}{\longrightarrow} \bZ_2^{D_Z},
\end{align}
where $d_0^b=d_0(K^b)^T$ and $d_1^b = d_1\oplus (\widetilde{M}^b)^T$, under slight abuse of notation.
$\widetilde{M}^b$ here denotes the generator matrix of twisted $Z$ errors, after appropriately adapting its domain from $\bZ_2^{D_X}$ to one choice of generators for the complement of its kernel.
Note that for each detector that is removed by $K^b$, one twisted error, i.e. a boundary, is added.
This keeps the rank of all homology groups the same.
Operationally, this guarantees that the linear operator $C_{b,c}$ teleports the same number of logical qubits from input to output as $C_{0,0}$.

\myparagraph{Reasons and implications of condition C}
The matrix $M^b$ is defined for any $b$ that is ``contractible'', in the sense that it is not close to any non-trivial 1-cycle, when considering $d_1$ and the connectivity of gates defining $S$ to define ``closeness''.
$M^b$ defines the abstract space of twisted detectors via its kernel.
Condition C is an additional condition on the existence of a trivial filling and is needed to calculate $\kappa_b$, the true values of the twisted detectors.
Moreover, it guarantees that any twisted error, as characterized by $\widetilde{M}^b$ in Thm.~\ref{thm:twistedconstraints}, leads to the same decoding outcome when mapping back the decoding output of the twisted $Z$ decoding problem to the untwisted one.

Condition C guarantees that the logical action of $C_{b,c}$ is related to that of $C_{0,0}$.
We understand it to be a contractibility condition on the filling itself.
Having a non-contractible filling means that the logical unitary is ill-defined over the ensemble of $c$ defined via all $Z$-like errors with the correct syndrome.
Moreover, the non-trivial homotopy class of the filling itself might result in a diagonal logical error which is independent of the homology class of $c$.
For any fixed filling the associated logical diagonal operator can be efficiently computed, by evaluating the path integral, i.e. calculating $S[A]$, but the circuit will sum over all relative cohomology classes of fillings.
This makes it particularly hard to assess the precise logical action in this setting, without further restricting the input state, to e.g. a logical computational basis state.

One can view the twisted errors as a mechanism to identify if a given correction in the twisted decoding problem introduced a logical error or not, by picking a reference class within the untwisted $Z$ problem.
This identification of homology classes within the two decoding problems is only canonical for a contractible filling, giving another operational meaning to the failure point of our framework if no contractible filling is found.

\myparagraph{Output states}
Consider a third-order circuit without any input edges.
The coefficients of its output state can be written as
\begin{align}\label{eq:third-order-state}
    \langle a\ket{\psi_{b,c}} = \sum_{A\colon d_1A=b}\delta_{A|_{\mathrm{out}}, a} e^{2\pi i S[A]}(-1)^{\langle c, A\rangle}.
\end{align}
The linear constraint in the sum induces an affine linear constraint on the computational basis states on which $\ket{\psi_{b,c}}$ is supported.
This is captured by a set of Pauli-$Z$ stabilizers that leave the state invariant, up to phase.
The associated parity-check matrix $d_1|_{\mathrm{out}}$ is obtained from puncturing the classical code defined by $d_1$ onto $T_{Z,\mathrm{out}}$, the space spanned by the $Z$ tensors on the output edges.
This means that $d_1|_{\mathrm{out}}$ is a generator matrix for the linear classical code formed by all elements in $\im(d_0)$ to $T_{Z,\mathrm{out}}$.
We can also extract the eigenvalue of each generator by truncating $b$ accordingly.

The output state $\ket{\psi_{b,c}}$ is also stabilized by Clifford operators.
These are of the form $DP_X$, where $D$ is a diagonal Clifford operator and $P_X$ a Pauli-$X$ operator, which we call \emph{twisted $X$ stabilizers}.
For $b=0$ the group is isomorphic to the image of $d_0$, truncated onto $T_{Z,\mathrm{out}}$. Denote this truncated operator by $d_0|_{\mathrm{out}}$.
For each element in that group, there is an associated $v\in\bZ_2^{D_X}$ for which $d_0|_{\mathrm{out}}v$ indexes the support of the $P_X$ part.
The diagonal operator in each stabilizer can be computed via Eq.~\eqref{eq:third-order-state}.
The matrix elements of the operator $D_v$, attached to a Pauli-$X$ with support $d_0|_{\mathrm{out}}v$ are given by
\begin{equation}
\begin{split}
    \bra{a} D_v\ket{a} =& \langle a\ket{\psi_{b,c}}\bra{a}\prod_{i\in d_0|_{\mathrm{out}}v}X_i\ket{\psi_{0,c}}^\ast\\
    =& (-1)^{\langle d_0^T c,v\rangle }e^{2\pi i(\Delta^{(2)}S)(A, d_0 v)},
\end{split}
\end{equation}    
where $A|_{\mathrm{out}}=a$ and $\bullet^\ast$ denotes complex conjugation.

Note that this defines a Clifford operator since $(\Delta^{(2)}S)$ is a second-order function, see App.~\ref{app:ithorder}.
This expression is independent of the explicitly chosen $A$ since $A$ is fixed on the boundary and $S$ is invariant under $A\mapsto A+ d_0v$, where $d_0v$ has no support on $T_{Z,\mathrm{out}}$.
A similar treatment of the output state of a post-selected third-order circuit can be found in Ref.~\cite[Sec. 3.1]{bauer2026gates}.

For non-trivial $b$, the stabilizers defined above actually do not commute since they are defined by wave-function coefficients on cocycles.
They commute up to Pauli $Z$ operators, defined by an affine linear function on the space of twisted $X$ stabilizers, given by
\begin{align}\label{eq:boundary-commutation}
\phi_v(a) - \phi_w(a) + \phi_w(a+ d_0|_{\mathrm{out}}v) -\phi_v(a+ d_0|_{\mathrm{out}}w),
\end{align}
where $\phi_v(a) = \langle d_0^T c,v\rangle/2 + 2\pi i((\Delta^{(2)}S)(A, d_0 v) - S(d_0 v))$ is the function characterizing the matrix elements of $D_v$.
This expression can be simplified by identifying $(\Delta^{(3)}S)$, showing some similarity to the bulk $M^b$ matrix.
Similar to the bulk twisted detectors, $\ket{\psi_{b,c}}$ is in a deterministic eigenstate of stabilizers that correspond to the kernel of the linear part of Eq.~\eqref{eq:boundary-commutation}.

\section{Simulating the decoding problem}\label{sec:decoding-sim}

We present an algorithm that faithfully generates a sample of the measurement statistics of a third-order circuit with a sampled Pauli error when condition C, defined in Def.~\ref{def:condition-C}, is fulfilled.
Without condition C we cannot efficiently characterize the syndrome distribution and also not guarantee the correct logical action of the non-Clifford block, so we declare decoding failure.
This leads to an overestimate in logical failure rates that are expected to be small for fault tolerant logical blocks.

The simulator builds on the technical results in Sec.~\ref{sec:FT-thirdorder}.
The $M$-matrices defined in Def.~\ref{def:Mmatrices} play a central role.

\subsection{Overview and subroutines}
We give an overview of the algorithm in Alg.~\ref{alg:overview}.
Since the non-Clifford gates do not affect the $Z$ errors directly, we assume that the $X$ and $Z$ errors in that circuit are decoded separately.
The $X$ errors might require a just-in-time decoder.
We denote the $X$ decoder with $\DecX\colon \bZ_2^{D_Z}\to \bZ_2^{T_X}$, the $Z$ decoder with $\DecZ\colon \bZ_2^{D_X}\to \bZ_2^{T_Z}$.
Since the $Z$ decoder can be run after the entire circuit was executed, it might get access to the $Z$ measurements $b_m$ and the $X$ correction $b_c$ to mitigate some of the twisted errors~\cite{highthresholdJIT}.
Moreover, we use subroutines, $\CheckCondC, \GenTwisted$ and $\GenLinking$, that we define in Algs.~\ref{alg:conditionC}, \ref{alg:twistederror} and \ref{alg:linkingerror}.

\begin{algorithm}
\caption{Estimate failure rates of third-order circuits with Pauli errors}
\label{alg:overview}
\DontPrintSemicolon
\LinesNumbered

\KwIn{Circuit description with complex Eq~\eqref{eq:faultcomplex-full}, Pauli noise model $\mathcal{N}$}
\KwOut{Sample of logical success or failure of third-order circuit with sampled errors}
Sample Pauli error $(b_e,c_e)\in \bZ_2^{T_X}\oplus \bZ_2^{T_Z}$ from $\mathcal{N}$\;
Run $\DecX(d_2b_e)\to b_c$\;
$b_e + b_c\to b$\;
\eIf{$b\in \im(d_1)$}{
$\CheckCondC(b)\to (\mathrm{check}, A)$\;
\lIf{$\mathrm{check}$}{
\Return{$\mathrm{failure}$}}}{
\Return{$\mathrm{failure}$}
}
$\GenTwisted(b, A) \to c_{\mathrm{tw}}$\;
$\GenLinking(b, A) + c_e +c_{\mathrm{tw}} \to c_e$\;
$\DecZ(d_0^Tc_e) \to \Tilde{c}_e$\;
\eIf{$c_e + c_m + c_{\mathrm{tw}}+ \tilde{c}_e\in \im(d_1^T)$}{\Return{$\mathrm{success}$}}{\Return{$\mathrm{failure}$}}
\end{algorithm}

In the following, we describe important, simple subroutines in the simulation algorithm, and explain how to deal with open boundaries, that interface with Clifford logical blocks.

\myparagraph{Sampling measurement outcomes}
Individual outcomes of third-order circuits might be random, e.g. in dimensional-jump protocols~\cite{bombin2016jump, bombin2018JIT, Brown2020JIT, Scruby2022JIT} and dynamical Floquet-style versions of non-Abelian protocols~\cite{bauer2025planar}.
In these cases, however, the decoding problem can be reliably simulated by setting all of the intrinsically random outcomes to $=1$ at the start of Alg.~\ref{alg:overview}.
The detector values that are used in decoding can be obtained from the errors directly.

A subtlety arises in non-Clifford blocks that implement a non-unitary logic gate and whose measurement outcomes themselves depend on the logical states.
Generically, this part of the syndrome distribution cannot be sampled from efficiently.
Concretely, only marginals of the measurement distributions will be subject to affine constraints.
In practice, these logical blocks are used for specific types of logical operations where we have prior knowledge about the distribution of logical outcomes.
For example, in many state-preparation protocols~\cite{tiedinknots, williamson2026fastmagicstatepreparation} or measurements used for teleportation or injection~\cite{tiedinknots, huang2026hybridlattice, manjunath2026groupsurface}, the logical outcomes are known to be i.i.d distributed.
This knowledge can be used to sample the entire distribution of physical measurement outcomes.
Additionally, the logical outcome will enter the decoding problem for $c$.
Formally, the chain complex in the last step in Alg.~\ref{alg:overview} is engineered to have a (high-weight) non-trivial cocycle for each logical outcome, each of which has its corresponding 1-cycle which corresponds to a logical fault.
We expand on how a simulator of entire logical circuits deals with fault tolerant measurement gadgets in Ref.~\cite{companion}.

\begin{remark}
    While the decoding problem can be modelled by setting all random measurement outcomes to $+1$, in a physical implementation the outcomes will be random and the physical modifications to the circuit, based on the decoder output, will depend on the outcomes explicitly.
    Concretely, Pauli-$X$ corrections need to be applied along a \emph{gauge fixing membrane} that fills $b$. Finding one filling will always succeed and the success of the protocol will not depend on which filling is chosen.
    This is analogous to corrections to non-trivial stabilizer measurements in dimensional-jump protocols~\cite{bombin2016jump}.
    Ref.~\cite{bauer2025planar} provides an explicit explanation of this phenomenon in third-order circuits that implement a non-Clifford logical block on the 2D color code.
\end{remark}

\myparagraph{Puncturing and shortening}
As a subroutine to $\CheckCondC$ we need to \emph{puncture} a length-2 chain complex
\begin{align}
    \bZ_2^{D_X} \to \bZ_2^{T_Z} \to \bZ_2^{T_X}
\end{align}
onto a subset $\Imp\subseteq T_Z$, in a specific basis.
For a CSS quantum code described by this chain complex this would mean to measure out the qubits in the complement of $\Imp$ in a specific complete CSS basis.

For $\CheckCondC$ we puncture in $X$.
As the first coboundary map, we consider the restriction of $d_0$ onto $\Imp$.
When interpreting $d_1$ as the generator matrix of a classical code, $d_1^{\Imp}$ is the matrix that generates the subspace of codewords that are truncated onto $\Imp_b$, i.e. the subspace obtained as the span when restricting \emph{all} codewords onto $\Imp$.
This procedure is also called \emph{puncturing}.
To obtain the second coboundary map, we consider the code with generator matrix $d_1^T$.
This code needs to be \emph{shortened} onto $\Imp$, by identifying a generating matrix $\partial_1^{\Imp}$ of $\im(d_1^T)\cap \bZ_2^{\Imp}$.
This involves a Gaussian elimination step which can be made relatively efficient and parallelized for sparse matrices and sparse $\Imp$.
Together, these two procedures define the punctured complex,
\begin{align}
    \bZ_2^{D_X^{\Imp}} \stackrel{d_0^{\Imp}}{\longrightarrow} \bZ_2^{\Imp_b} \stackrel{d_1^{\Imp}}{\longrightarrow} \bZ_2^{T_X^{\Imp}},
\end{align}
where $d_1^{\Imp} = (\partial_1^{\Imp})^T$.
In Alg.~\ref{alg:conditionC} both puncturing and shortening are separated into two subroutines.

\myparagraph{Checking condition C}
In the general case we propose to approximately check condition C by finding a minimum-weight solution to $d_1A = b$.
This is equivalent to the decoding problem of an LDPC code which can be efficiently approximated.
It can also be checked separately on every connected component of $b$ individually, making it less likely to fail for a high-weight but clustered configuration $b$.
Moreover, in large circuits, errors will be sufficiently sparse that it can be further parallelized.
For specific instances, we expect that partial analytic understanding of condition C can help to further reduce the error in approximating checking condition C.

\begin{algorithm}
\caption{$\CheckCondC$: Check condition C (approximately)}
\label{alg:conditionC}
\DontPrintSemicolon
\LinesNumbered
\KwIn{2-cocycle $b\in{\im(d_2)}$, efficient approximator from min-weight solution of $d_1A =0$ $\mathrm{ApproxDec}$}
\KwOut{Found no contractible filling $\in \{\mathrm{true}, \mathrm{false}\}$, a filling $A$}
$\mathrm{ApproxDec}(b)\to A$\;
Puncture $d_0^T$ onto $\Imp_A$, obtain $\partial^{\Imp}\colon \bZ_2^{\Imp_A} \to \bZ_2^{D_X^{\Imp}}$\;
Shorten $d_1$ onto $\Imp_A$, obtain $d^{\Imp}\colon \bZ_2^{\Imp_A} \to \bZ_2^{D_Z^{\Imp}}$\;
$\bigg[\bZ_2^{D_X^{\Imp}} \stackrel{(\partial^{\Imp})^T}{\longrightarrow}\bZ_2^{\Imp_A} \stackrel{d^{\Imp}}{\longrightarrow} \bZ_2^{D_Z^{\Imp}} \bigg]\to C$\;
\eIf{$\dim(H_1(C))=0$}{
\Return{$(\mathrm{true}, A)$}}{
\Return{$(\mathrm{false}, 0)$}}
\end{algorithm}

\begin{algorithm}
\caption{$\GenTwisted$: Sample twisted error}
\label{alg:twistederror}
\DontPrintSemicolon
\LinesNumbered
\KwIn{contractible 2-cocycle $b$, contractible filling $A\colon d_1A=b$}
\KwOut{random $c_{\mathrm{tw}}\in \mathcal{E}_Z^b$}
Calculate $\widetilde{M}^b(A,\bullet, \bullet)$ (Def.~\ref{def:Mmatrices})\;
$0\to c_{\mathrm{tw}}\in \bZ_2^{T_Z}$\;
\ForEach{$v\in D_X$}{
Sample random bit $x\in\{0,1\}$\;
\If{$x=1$}{$c_{\mathrm{tw}}\to c_{\mathrm{tw}} + \widetilde{M}^b v$}
}
\Return{$c_{\mathrm{tw}}$}
\end{algorithm}

\begin{algorithm}
\caption{$\GenLinking$: Calculate offset-error}
\label{alg:linkingerror}
\DontPrintSemicolon
\LinesNumbered
\KwIn{contractible 2-cocycle $b$, contractible filling $A$}
\KwOut{offset error $c_0\colon \partial_0^b c_0 = \kappa_b$}
Calculate $M^b(A, \bullet, \bullet)$\;
Calculate $K^b$, injective generator matrix of $\ker(M^b)$\;
$2(\Delta^{(2)}S)(A, (K^b)^T \bullet) \to \kappa_b$\;
Find one solution $c_0$ of $K^b d_0^T c_0 = \kappa_b$\;
\Return{$c_0$}
\end{algorithm}

\begin{remark}
    In the description of the algorithms above we have left many subroutines implicit.
    These steps all involve evaluating equations over $\bZ_2$ or solving linear equations over $\bZ_2$ which can be done efficiently.
    In practice, these subroutines can also be highly optimized and parallelized, given the locality structure of the circuits.
    We expect that our methods will work particularly well for large and structured circuits, for example circuits that admit a form of translation invariance, or more general group-translation symmetry.
    We leave explicit implementations and optimizations to future work.
\end{remark}

\begin{remark}
    Ref.~\cite{highthresholdJIT} demonstrates a numerical implementation of our formalism for the specific case of a topological protocol where the non-Clifford gates are defined via a triple-intersection cohomology invariant.
    It shows that the overestimate in logical failure rate can be bounded in structured circuits where condition C can be efficiently checked.
\end{remark}

\subsection{Interfacing with Clifford blocks}\label{sec:interface}

In this paragraph we discuss how to define the boundary conditions on the decoding graphs that enter the simulation algorithms for noisy input states and how to produce the necessary data to interface with Clifford simulators.

We consider an input state, encoded in a CSS code, subject to a correctable set of Pauli errors.
These are captured by an \emph{input Pauli frame}.
In the following we describe how to define the decoding graph with open output boundaries and how the errors and corrections can be mapped to an \emph{output Pauli frame} on a CSS stabilizer group with respect to which the next Clifford block is defined.

To simplify the boundary treatment, we assume that the input and output stabilizer codes on which the third-order circuit implements a logic gate are CSS and we explicitly incorporate the last round of syndrome extraction of the input and the first round of syndrome extraction for the output code into the non-Clifford block.
If the Clifford circuits before and/or after the non-Clifford block are dynamical protocols that and not related to a single CSS stabilizer code, we include the part of the circuit that reads out, respectively initializes, a CSS stabilizer group with extensive distance, to have a complete set of half-open CSS detectors on both boundaries.

The task of the simulator for the block is to take an input Pauli frame, perform the simulation of the resulting decoding problem for the non-Clifford block and, conditioned on decoding success, outputs an output Pauli frame.
In our setting decoding success implies that $b$ fulfills condition C.
The output Pauli frame depends explicitly on the output of $X$ decoding within the non-Clifford block, which defines the 2-cocycle $b$.
Note that the $Z$ decoding is not strictly needed to update the Pauli frame as no active $Z$ corrections are needed within the non-Clifford block to be fault tolerant.

The input frame is described by a binary string $(e_x^{\mathrm{out}},e_z^{\mathrm{in}})\in\bZ_2^{T_{Z,\mathrm{in}}}\oplus\bZ_2^{T_{Z,\mathrm{in}}}$, labelling the input Pauli errors in their symplectic representation.
We use simple tensor-network rewrite rules, Eqs.~\eqref{eq:Ztensor-symmetry}, \eqref{eq:Xtensor-symmetry} and \eqref{eq:fusionrules-tensors} to find the equivalent vectors in $\bZ_2^{T_Z}$ and $\bZ_2^{T_X}$.
The $Z$ part can be directly fused into the boundary $Z$ tensor via the natural embedding $\bZ_2^{T_{Z,\mathrm{in}}}\hookrightarrow\bZ_2^{T_Z}$.
Since there is no third-order tensor directly connected to the input $Z$ tensors, we can push the $X$ tensor through, using Eq.~\eqref{eq:Ztensor-symmetry}, where it splits and maps onto $d_1e_x\in\bZ_2^{T_X}$.
This defines the embedding of $e_x$ into $\bZ_2^{T_X}$.
In this way we have defined the \emph{closed boundary condition} to the decoding problem for the third-order circuit.
At this boundary, the detector values are fixed from the values from the previous (Clifford) block.
Since we included a complete set of open CSS detectors via padding with a shallow CSS circuit there is a canonical identification of the detectors from the previous round to fixed affine constraints on the input boundary of the non-Clifford block.

The output boundaries to the decoding graphs of the non-Clifford block are left \emph{open}, i.e. $b,c$ are allowed to terminate into the future-directed boundary.
This can be achieved for example, by treating the half-open boundary detectors as actual detectors and setting the weight of all (hyper)edges that include them to 0.
The termination pattern of $b,c$ on these detectors will define the syndrome of the output state with respect to the output stabilizer group.
To find the correct Pauli frame the $b$ and $c$ configurations in the bulk of the non-Clifford block are propagated to the output boundary.
This can be done directly in the underlying CSS circuit, using tensor-network rewrites or, equivalently, standard Pauli frame tracking tools based on the stabilizer formalism.

\section{Outlook}\label{sec:outlook}
We have proven that the measurement outcomes within fault tolerant third-order circuits can be efficiently sampled from if the $X$ errors have trivial logical effect.
This guarantees that the measurement distribution itself does not depend on the logical information flowing through the circuit.
We show that the resulting distribution of measurement outcomes is constrained by linear constraints that can be interpreted as \emph{detectors in a non-Clifford circuit}.
We use these constraints to construct an efficient sampling procedure of the measurement outcomes in the circuit and to derive a well-defined decoding problem.
Moreover, we have derived sufficient conditions to determine if a given error-configuration induces a non-trivial logical error or not.
Taken together, this provides a framework to efficiently estimate the failure rates of non-Clifford logical blocks where the error in estimation is small for fault tolerant circuits.

Identifying the precise relationship of the linear constraints in our spacetime perspective and the Clifford stabilizers in a stroboscopic view of the circuit from Ref.~\cite{yugopaper} could allow for efficient simulation and decoding of protocols with low fault distance.

Our framework has been numerically implemented for a topological circuit in Ref.~\cite{highthresholdJIT}.
It would be important to obtain more realistic estimates on its failure rate under circuit-level noise, and apply it to other protocols that promise more efficient magic-state preparation schemes~\cite{zhu2026nonabelianqldpc, christos2026gauging, williamson2026fastmagicstatepreparation}.

Going forward, our formalism allows us to interface with simulators of fault-tolerant Clifford circuits that are based on tracking Pauli frames.
This will allow us to efficiently estimate failure rates of entire fault tolerant implementations of algorithms or their subroutines~\cite{companion}.

It would be interesting to investigate how to extend the formalism to adaptivity in the physical circuits.
This could potentially allow us to generalize schemes that reduce the physical overhead for Clifford circuits to the non-Clifford setting~\cite{Berthusen2025adaptive, birchall2026macromux, bartolucci2025comparison}.

Generalizing detectors in non-Clifford circuits beyond third-order circuits would be important to quantify error-correcting properties of larger classes of circuits.
It is possible to represent circuits that also include other types of gates as a path integral.
For example, Hadamard gates can be represented using an auxiliary variable~\cite{bauer2026gates}.
It is also conceivable that parts of the formalism can be extended to diagonal gates in higher levels of the Clifford hierarchy, or mixed-dimensional systems.
Obtaining a quantitative understanding of the resulting symmetries of the path-integral would allow us to derive emerging linear constraints and help to understand fundamental limitations to capture the measurement statistics of a fault tolerant circuit.

In a broader context it would be interesting to obtain a similar homological interpretation of fault-tolerant blocks that operate on non-CSS codes, using 3-term fault complexes, as used in Refs.~\cite{pesah2025spacetimecode, yuan2026noncss}.
This could lead to more efficient schemes to implement fault tolerant logic on general LDPC codes.

\section*{Acknowledgements}
The framework and algorithms presented in this work have been derived without the help of large language models.
GPT-5.6 Sol and GPT-6 Astra were used to streamline the writing and for parts of the calculation of the examples presented in App.~\ref{app:examples}.

We thank Y. Takada, S. Bartlett and D. Williamson for sharing their work on Ref.~\cite{yugopaper} with us.
JM thanks A. Bauer for previous collaborations and fruitful discussions.
Part of this work was done while JM and TS attended the FTQT workshop at the Centro de Ciencias de Benasque Pedro Pascual. 
JM acknowledges funding from the DFG (CRC183) and from Defence Science and Technologies Group (DSTG) and Advanced Strategic Capabilities Accelerator (ASCA) through its Emerging and Disruptive Technologies (EDT) Program.
This work was supported in part by Japan Science and Technology Agency (JST) as part of Adopting Sustainable Partnerships for Innovative Research Ecosystem (ASPIRE), Grant Number JPMJAP25A3.

\onecolumngrid

\begin{appendix}

\section{$i$th order functions and diagonal gates in the Clifford hierarchy}\label{app:ithorder}
Any diagonal gate on $n$ qubits is determined by its matrix elements.
These are captured by a function $S:\bZ_2^{n} \to \bR/\bZ\simeq [0,1)$,
\begin{align}
    S[A] = \sum_{x\in \bZ_2^n} S_x \prod_{i = 1}^{n} A_{x_i},
\end{align}
and the corresponding gate can be written as $D_S = \sum_A e^{2\pi i S[A]} \ketbra{A}$.

Ref.~\cite{Cui2017diagonal} already points out that the functions that define diagonal gates in the Clifford hierarchy~\cite{Gottesman1999} are highly constrained and described by so-called \emph{phase polynomials}.
In the following, we shortly review an equivalent formulation in terms of \emph{$i$th-order functions} between Abelian groups which we use to derive linear constraints in the circuit that allow for an efficient description and efficient sampling method of the correct Pauli-error distribution.

\begin{definition}[see also Def. 32 in Ref.~\cite{bauer2026quadratictensors}]\label{def:ithorder}
Let $G,M$ be two Abelian groups and $f:G\to M$ a function from $G$ to $M$.
We introduce the \emph{$i$th derivative} of $f$ recursively as,
\begin{align}\label{eq:def-derivative}
\begin{split}
    \Delta^{(i)}f: G^{\times i} \to& M,\\
    (\Delta^{(i)}f)(g_1,...,g_i) =& (\Delta^{(i-1)}f)(g_1,g_3,...,g_i) +_M (\Delta^{(i-1)}f)(g_2,g_3,...,g_i)\\
    &-_M (\Delta^{(i-1)}f)(g_1 +_G g_2 ,g_3,...,g_i),
\end{split}
\end{align}
where ``$+_G$'' denotes the group operation in $G$, and similarly for $M$.
The first two derivatives are defined as
\begin{align}
    \Delta^{(0)}f = f(0) \qq{and} (\Delta^{(1)}f)(g) = f(g)-f(0)
\end{align}

We call a function $f$ an \emph{$i$th order function} if its $(i+1)$th derivative vanishes,
\begin{align}\label{eq:def-ithorder}
    \Delta^{(i+1)}f = 0.
\end{align}
We denote the set of all $i$th order functions from $G$ to $M$ with $\mathcal{F}_i(G,M)$. If clear from context, we omit $G$ and $M$.
\end{definition}

Note that Eq.~\eqref{eq:def-derivative} seems to treat the first argument of $\Delta^{(i+1)}f$ in a special way.
It turns out, in fact, that the resulting functions, however, do not depend on that choice, and it is hence sufficient to consider changes of $\Delta^{(i+1)}f$ in the first component when defining the derivative.

\begin{corollary}\label{cor:derivates-symmetric}
    Let $f:G\to M$ be a function from $G$ to $M$.
    Any derivative is symmetric in its arguments, $(\Delta^{(i)}f)(g_1, g_2, ...,g_i) = (\Delta^{(i)}f)(\pi(g_1), \pi(g_2), ...,\pi(g_i))$ for any permutation $\pi$ on $i$ elements.
    It follows that the $i$th derivative of an $i$th order function is a group homomorphism on $G^{\times i}$.
\end{corollary}

\begin{proof}
    Can be proven by induction.
    For the base case, consider $\Delta^{(2)}f: G^{\times 2} \to M$. It is defined as
    \begin{align*}
        (\Delta^{(2)}f)(g_1,g_2) = (\Delta^{(1)}f)(g_1) +_M (\Delta^{(1)}f)(g_2) -_M (\Delta^{(1)}f)(g_1+_Gg_2),
    \end{align*}
    which is obviously symmetric.
    Now, assume that $\Delta^{(i)}f$ is symmetric. It follows directly that $\Delta^{(i+1)}f$ is invariant under any permutation of the arguments $g_3, ..., g_i$. Symmetry in the other arguments can be inferred from the definition of $\Delta^{(i)}f$ using commutativity of group multiplication in both $G$ and $M$.

    The second statement in the Corollary follows directly from Eq.~\eqref{eq:def-ithorder} and the symmetry that was just shown for $\Delta^{(i)}f$.
\end{proof}

There is a direct correspondence between diagonal gates in the Clifford hierarchy~\cite{Cui2017diagonal} and $i$th order functions for $G=\bZ_2^n$ and $M=\bR/\bZ\simeq U(1)$.

\begin{definition}[Clifford hierarchy]
    Let $\mathcal{P}_n$ be the Pauli group on $n$ qubits. 
    The Clifford hierarchy
    \begin{align}
        C_1 = \mathcal{P}_n \subseteq C_2 \subseteq C_3 \subseteq ...
    \end{align}
    is defined recursively by
    \begin{align}
        C_i = \{U\in U(2^n)\;|\; UPU^\dagger \in C_{i-1} \;\forall P\in \mathcal{P}_n\}.
    \end{align}
\end{definition}

\begin{proposition}(Diagonal gates in the Clifford hierarchy)
    Let $D_i\subseteq C_i$ be the set of diagonal gates in the $i$th level of the Clifford hierarchy.
    The following two properties hold
    \begin{itemize}
        \item $D\in D_i \Leftrightarrow \bigg\{\comm{\comm{\comm{\comm{D}{X^{(1)}}}{X^{(2)}}}{...}}{X^{(i)}} \propto \one$
        for any collection of Pauli $X$ operators $\{X^{(j)}\}_{j=1}^i \bigg\}$.
        \item $D_i$ is a group
    \end{itemize}
\end{proposition}

\begin{proof}(sketch)
    The first property can be proven using that a general $n$-qubit Pauli operator $P$ can be written as $P=i^{p_{\phi}} Z^{p_z}X^{p_x}$ where $p_{\phi}\in \{0,1,2,3\}$ defines a global phase and $p_z,p_x\in \bZ_2^n$ denote the support of the $Z$ and $X$ component of $P$.
    Since $D_i$ is diagonal, the group commutator $\comm{D_i}{P}$ only depends on $p_x$. Moreover, it is diagonal.
    Taken together, it follows that
    \begin{align}
        \comm{D_i}{P} \in D_{i-1}.
    \end{align}
    Iterating this argument shows that $\comm{\comm{\comm{\comm{D}{X^{(1)}}}{X^{(2)}}}{...}}{X^{(i-1)}}\in C_1$, i.e. is a diagonal Pauli operator, and hence, must either commute or anti-commute with any Pauli $X$ operator $X^{(i)}$.
    The argument can be inverted to show the equivalence.
    Since the only operators that commute up to a phase with any Pauli $X$ operator are diagonal Pauli operators $\comm{\comm{\comm{\comm{D}{X^{(1)}}}{X^{(2)}}}{...}}{X^{(i-1)}} \stackrel{!}{\in} C_1$.
    By definition of $C_2$, it follows that $\comm{\comm{\comm{\comm{D}{X^{(1)}}}{X^{(2)}}}{...}}{X^{(i-2)}} \stackrel{!}{\in} C_2$.
    This argument can be iterated and we find that $D\in C_i$.

    The second property follows from $d_1,d_2\in D_i \implies d_1d_2\in D_i$ which can be verified using the elementary properties of group commutators, $[AB, C] = A[B, C]A^{-1}[A, C]$, together with the first property.
\end{proof}

\begin{lemma}
    Let $X^{x}\in \langle \{X\}_{i=1}^n\rangle$ be a Pauli-$X$ operator which we identify with a vector $x\in \bZ_2^n$, $X^x\ket{A} = \ket{A+x}$, for any $A\in \bZ_2^n$.
    The gate $D_{f_i}$ defined by an $i$th order function $f_i:\bZ_2^n\to \bR/\bZ$ with $f_i(0)=0$ is in the $i$th level of the Clifford hierarchy, i.e. for any $x\in \bZ_2^n$
    \begin{align}
        \comm{D_{f_i}}{X^{x}} = X^{x}D_{f_i} X^x D_{f_i}^\dagger = D_{f_{i-1}},
    \end{align}
    for an $(i-1)$th order function $f_{i-1}$ that depends on $x$.
    Vice versa, any diagonal gate in the $i$th level of the Clifford hierarchy is of the form $D_{f_i}$ for an $i$th order function $f_i$.
\end{lemma}

\begin{proof}
    To prove the equivalence we prove two correspondences: Any $i$th order function defines a gate in $C_i$ and the matrix elements of any diagonal gate in $C_i$ are given by an $i$th order function.
    
    Since $D_{f_i}$ is diagonal, we can prove the first statement by showing that
    \begin{align}\label{eq:nested-comm-ithlevel}
        \comm{\comm{\comm{\comm{D_{f_i}}{X^{x_1}}}{X^{x_2}}}{...}}{X^{x_i}} = e^{2\pi i g(\{x_j\}_j)} \one,
    \end{align}
     for any $x = (x_1,...,x_i)\in \bZ_2^{\times i}$ and a function $g:(\bZ_2^n)^{\times i} \to \bR/\bZ$ derived from $f$. We first consider an arbitrary function $f:\bZ_2^n\to \bR/\bZ$ and the associated diagonal operator $D_f$.
    We calculate iteratively,
    \begin{subequations}\label{eqs:nested-comms}
    \begin{align}
        \begin{split}
            \comm{D_f}{X^{x_1}} =& \sum_{A} e^{2\pi i( f(A) - f(A+x_1)) } \ketbra{A}\\
            =& e^{-2\pi i f(x_1)} \sum_{A} e^{2\pi i (\Delta^{(2)}f)(A, x_1)} \ketbra{A}  \qcomma
        \end{split}\\
        \begin{split}
            \comm{\comm{D_f}{X^{x_1}}}{X^{x_2}} =& e^{-2\pi i f(x_1)} \sum_A e^{2\pi i( (\Delta^{(2)} f)(A,x_1) - (\Delta^{(2)}f)(A+x_2,x_1))}\ketbra{A}\\
            =& e^{-2\pi i (f(x_1) + (\Delta^{(2)}f)(x_2,x_1))} \sum_A e^{2\pi i (\Delta^{(3)}f)(A,x_1,x_2)}\ketbra{A},
        \end{split}\\
     \smash{\vdots} \notag \\
    \begin{split}\label{eq:nested-comm-general}
        \comm{\comm{\comm{\comm{D_{f}}{X^{x_1}}}{X^{x_2}}}{...}}{X^{x_i}} =& e^{-2\pi i \left(\sum_{j=1}^{i} (\Delta^{(j)}f)(x_1, ..., x_j) \right)} \sum_A e^{2\pi i (\Delta^{(i+1)}f)(A,x_1,x_2, ..., x_i)}\ketbra{A}
    \end{split}
    \end{align}
    \end{subequations}
    If $f\in\mathcal{F}_i$, $\Delta^{(i+1)}f = 0$ and the right hand side is proportional to $\one$ and we identify
    \begin{align}
        g(x_1, x_2, ..., x_i) = \sum_{j=1}^{i} (\Delta^{(j)}f)(x_1, ..., x_j).
    \end{align}

    To prove the opposite direction note that \eqref{eq:nested-comm-ithlevel} holds for any diagonal gate in the $i$th level of the Clifford hierarchy. It then follows from Eq.~\eqref{eq:nested-comm-general} that the function that defines the matrix elements of the gate must satisfy $\Delta^{(i+1)}f = 0$, i.e. it is an $i$th order function.
\end{proof}

Given the correspondence, we refer to diagonal gates in the $i$th level in the Clifford hierarchy as $i$th order gates.

\section{Examples to calculate twisted errors}\label{app:examples}

In this appendix we demonstrate the formalism described in this paper with three pedagogical examples where errors are assumed to be applied in a single time step, on an otherwise perfect code state.
These errors can be considered the errors obtained as a residual correction from a previous round of correction, prior to applying the non-Clifford gate.
In this setting the computation suits itself to an analytical treatment, and represents the calculation of phenomenological noise in an error-correcting circuit whose fault complex is described by the complex obtained from the tensor product between the complex of the CSS code analysed here and a length-2 repetition code, representing the initial and final measurements of the stabilizers.

\myparagraph{Transversal $CCZ$}
Consider a triple of LDPC CSS codes on the same number of physical qubits, where transversal application of $CCZ$ between the code blocks is a logic gate.
The transversal gate is defined by two bijections $Q_1\stackrel{\varphi}{\rightleftharpoons} Q_2 \stackrel{\vartheta}{\rightleftharpoons} Q_3$, between the qubits of the three code blocks, denoted by $Q_1, Q_2, Q_3$.
The bijections define the triples of qubits coupled by the $CCZ$, $\{\vb{i} = (i, \varphi(i), \vartheta(\varphi(i)))\}_{i\in Q_1}$.
The associated third-order function is
\begin{align}
    S[A] = \frac{1}{2} \sum_{\vb{i}} \overline{A}_{\vb{i}_1} \overline{A}_{\vb{i}_2}\overline{A}_{\vb{i}_3},
\end{align}
where $A=(A_1, A_2, A_3)\in\bZ_2^{Q_1}\oplus \bZ_2^{Q_2}\oplus \bZ_2^{Q_3}$ is the vector spanned by all qubits among the three code blocks.
The index $\vb{i}_1$ refers to the $i$th component in the first block, $\vb{i}_2$ the $\varphi(i)$th component in the second block and similarly for $\vb{i}_3$.
$S$ is understood modulo $1$, and $\overline{\bullet}$ denotes the lift from $\{0,1\}\to \bZ$.

For simplicity, consider the setting where a code state, that can be thought of as the output of a perfectly post-selected measurement of the code-stabilizers, is subject to some $X$ error in a correctable region of the code. In practice, this will be a residual error of a QEC correction step just before applying the gate.
We aim to characterize the output distribution of stabilizer measurements after applying the gate.
The $Z$-stabilizer measurements are determined by the residual $X$ error and additional measurement errors, both of which are captured by $b$, a binary string on the $X$ measurement-error locations prior to applying the transversal gate.
Since the $X$ error acts on a correctable region of the code, we can choose a filling $A$ that is supported on the $Z$ tensors that are on the worldlines of the qubits on which the error acted, starting from the time-slice after the error happened.
This is a special case of the case we consider here, consisting of only two rounds of stabilizer measurements with phenomenological noise.
Without further $Z$ errors, some $X$ stabilizers are still violated and their pattern is fully captured by the twisted errors.

To calculate $\widetilde{M}^b$, we identify $\bZ_2^{D_X}$ with the detectors that are formed by the initial $X$ stabilizer and the consecutive measurement, so they are in 1-1 correspondence with $X$-stabilizers of the code.
For each measured stabilizer generator $v$, $d_0v$ is supported on the spacetime qubit locations of the qubits in the support of $v$.
First, consider a single stabilizer generator in code block 1, $d_0 v = ((d_0v)_1, 0,0)$, which we abbreviate without the subscript.
We calculate by expanding $\Delta^{(3)}S$ according to Def.~\ref{def:ithorder},
\begin{subequations}
\begin{align}
\widetilde{M}^b_{e,v}
&= 2(S[A+e+d_0v] - S[A+e] - S[A+d_0v] - S[e+d_0v] + S[A] + S[e] + S[d_0v])
\\
&= \sum_{\vb{i}}\Bigl(
   (\overline{A+e+d_0v})_{\vb{i}_1}
   (\overline{A+e})_{\vb{i}_2}
   (\overline{A+e})_{\vb{i}_3}
   -(\overline{A+e})_{\vb{i}_1}
    (\overline{A+e})_{\vb{i}_2}
    (\overline{A+e})_{\vb{i}_3}
\nonumber\\
&\qquad
   -(\overline{A+d_0v})_{\vb{i}_1}
    \overline{A}_{\vb{i}_2}
    \overline{A}_{\vb{i}_3}
   +\overline{A}_{\vb{i}_1}
    \overline{A}_{\vb{i}_2}
    \overline{A}_{\vb{i}_3}
    + \overline{e}_{\vb{i}_1}\overline{e}_{\vb{i}_2}\overline{e}_{\vb{i}_3} - (\overline{e+ d_0v})_{\vb{i}_1}\overline{e}_{\vb{i}_2}\overline{e}_{\vb{i}_3}\Bigr)
\\
&= \sum_{\vb{i}}
      d_0v_{\vb{i}_1}\left(
         A_{\vb{i}_2}e_{\vb{i}_3}
         +A_{\vb{i}_3}e_{\vb{i}_2}
      \right) \in \bZ_2.\label{eq:CCZtildeM-elements}
\end{align}
\end{subequations}
In the second step, we have used that the entire expression is interpreted modulo $2$ such that the individual lifts to integers can be combined to a single one.
Then, all expressions cancel that do not involve $d_0v$.
A similar expression holds for the stabilizers supported on code blocks 2 and 3, respectively, by permuting the indices.

From this expression, we can directly read off the twisted errors as the vectors $e$ for which Eq.\eqref{eq:CCZtildeM-elements} is non-zero.
For a given $A\colon d_1A=b$, the space of twisted errors associated to a single stabilizer supported on block 1 is spanned by products of $Z$ errors on block 2 and block 3.
It is non-trivial on the qubits of block 2 that are involved in a $CCZ$ between a qubit that is part of the stabilizer support $v$ on block 1 and a qubit in block 3 on which $A$ is non-zero, and similarly for $Z$ errors on code block 3.

More generally, the space of twisted errors can be expressed as
\begin{align}\label{eq:twistederrors-CCZ}
    \mathcal{E}_Z^b = \left\{ \mqty(x^2\odot A^3 + x^3\odot A^2\\ x^1\odot A^3 + x^3\odot A^1\\ x^1\odot A^2 + x^2\odot A^1) \;\Bigg|\; x=d_0v,\; v\in \bZ_2^{D_X}\right\},
\end{align}
where $\odot$ denotes pointwise multiplication between the vectors, $a\odot b = (a_1b_1, a_2b_2, ...)$.
To fully characterize the syndrome distribution, we need to calculate the expected values of the twisted detectors, $\kappa_b$.
For this, we explicitly need $(K^b)^T$, the generator matrix of $\ker(M_b)$.
Both quantities depend on $b$ explicitly and will define the subspace of detectors in $\bZ_2^{D_X}$ that is not violated by 
$\mathcal{E}_Z^b$.
Due to non-trivial cancellations in the untwisted syndrome of the twisted errors, we cannot quantify $K^b$ better without an explicit code and $b$.
Once given, computing $(K^b)^T$ is efficient and $\kappa_b$ can be computed as stated in Thm.~\ref{thm:twistedconstraints}.

\myparagraph{Transversal $T$}
In the same way as for the $CCZ$ we can analyse a transversal $T$ gate on a single CSS code block with physical qubits $Q$.
The associated third-order function is
\begin{align}
    S[A] = \sum_{i\in \hat{Q}} \frac{1}{8}\overline{A}_i,
\end{align}
where $\hat{Q}$ denotes the spacetime locations where the $T$ is supported, which is in bijection with the qubits of the CSS code.
As in the first example, the residual $X$ errors have syndrome $b$ in the circuit and we can choose a filling $A$ supported on the $Z$ tensors of the worldlines of the qubits on which the error acted.

To calculate $\widetilde{M}^b$, we identify the untwisted detectors $\bZ_2^{D_X}$ with the $X$-stabilizers of the code.
For each stabilizer generator $v$, its associated detecting region $d_0v$ intersects with the spacetime qubit locations associated with the support of $v$.
Using this, we explicitly calculate
\begin{subequations}
\begin{align}
    \widetilde{M}^b_{e,v} =& \frac{1}{4}\sum_i \bigg( (\overline{A+e+d_0v})_i - (\overline{A+e})_i - (\overline{A+d_0v})_i - (\overline{e+d_0v})_i + \overline{A}_i + \overline{e}_i + \overline{d_0v}_i \bigg)\\
    =& \sum_i \overline{A}_i\overline{e}_i\overline{d_0v}_i \in \bZ_2,
\end{align}
\end{subequations}
where the subscript $T$ indicates the restriction of the vectors onto the support of the physical $T$ gates in spacetime.
We used the identity for the triple product of integers that holds in $\bZ_2$,
\begin{align}
    \frac{1}{4}\left( \overline{a+e+x} - \overline{a+e} - \overline{a+x} - \overline{e+x} + \overline{a} + \overline{e} + \overline{x}\right) = aex \mod 2.
\end{align}
Since we are considering the phenomenological case, a simple choice of $A$ is the extension of $b$ into the time direction.
Hence, $A$ can be thought of as being supported on the worldlines of the qubits on which the $X$ errors acted.
Consequently, the twisted errors can again be expressed using the pointwise multiplication, and we obtain
\begin{align}\label{eq:twistederrors-T}
    \mathcal{E}_Z^b = \left\{ A_T\odot d_0v \;|\; v\in\bZ_2^{D_X}\right\},
\end{align}
where $A_T$ denotes the restriction of $A$ onto $\hat{Q}$.
Without further structure of the error and its syndrome $b$ we cannot quantify the number of independent twisted detectors and hence also not give an explicit formula for $\kappa_b$.
For any instance, however, we can efficiently find the exact expression algorithmically which is not that different from a computable algebraic expression.
The space of twisted detectors is the kernel of the matrix with elements
\begin{align}
    M^b_{v,w} = \sum_i (d_0 v_T)_i(A_T)_i(d_0w_T)_i = \langle d_0v_T, A_T d_0w_T\rangle.
\end{align}
And given $A$ we can compute $\kappa_b$ using the formula in Thm.~\ref{thm:twistedconstraints}.
These expressions have been derived in more detail in Ref.~\cite{bauer2025planar} for families of 3D color codes, on general 3-colexes.

\begin{remark}
    Note that a similar formula to Eqs.~\eqref{eq:twistederrors-CCZ} and \eqref{eq:twistederrors-T} holds for each $CCZ$ and $T$ individually, and similarly for $CS$, which leads to an efficient way to sample twisted errors, without the offset error that guarantees that the expected values are the correct ones, $\kappa_b$.
\end{remark}

\myparagraph{Cup-product circuit}
As a last example, we review the calculation from Ref.~\cite{highthresholdJIT}, for $X$ errors prior to a 3-copy cup product circuit~\cite{Breuckmann2026cupsgates}.
This calculation only assumes a residual $X$ error configuration before the circuit. Our formalism would also capture errors throughout the circuit. The expressions would then, however, not be expressed in terms of cup products alone but would depend on more fine-grained information on the circuit.

The cup-product on the CSS code gives a CSS circuit between triples of the code. 
We consider the qubits of the three code blocks in a single set, $Q=(Q_1,Q_2,Q_3)$.
For each qubit in a single block $q$, we denote its ``copy'' on block $i$ with $q^i$.
In that notation, the cup-product unitary gate can be represented by a polynomial
\begin{align}
    S[A] = \frac{1}{2} A^1\cup A^2 \cup A^3,
\end{align}
where $\hat{Q}$ denotes the spacetime locations of the unitary, which is in bijection with the set of qubits of a single block.

Similar to the prior examples, we consider a residual $X$ error configuration that defines a 2-cocycle $b$ subject to condition C.
In the simple case we consider here the filling can be chosen to be supported on the worldlines of the qubits on which the error acted, just after the $X$ errors.

The cup product is constructed to fulfil the Leibniz rule
\begin{align}
    d(x\cup y) = dx\cup y + x\cup dy
\end{align}
for arbitrary cochains $x,y$.
Here, $x,y$ are spacetime $Z$ error locations and supports of $CCZ$ gates defined by $S$.
This makes the calculation of twisted errors particularly tractable.
For their generating matrix we obtain
\begin{align}
    \widetilde{M}^b_{e,v} = e^1\cup b^2 \cup v^3 + b^1\cup e^2\cup v^3 + e^1\cup v^2 \cup b^3 + b^1\cup v^2 \cup e^3 + v^1\cup e^2 \cup b^3 + v^1\cup b^2 \cup e^3,
\end{align}
where we used the Leibniz rule and $d_1A = b$.
Note that this simply sums over all triples of $b^i, v^j, e^k$ for which the cup-product is non-zero.
A similar expression can be found for $M^b$ and $\kappa_b$.

\end{appendix}

\bibliographystyle{quantum}
\bibliography{refs}

\end{document}